%% file: main.tex
\documentclass[runningheads,envcountsame]{llncs}
\pdfoutput=1

\usepackage{silence}
\usepackage{etoolbox}
\apptocmd{\sloppy}{\hbadness 10000\relax}{}{} 

\usepackage{hyperref}

\usepackage{mathtools}
\usepackage{mathpartir}
\usepackage{amsfonts}
\usepackage{amssymb}
\usepackage{amsmath}
\usepackage{stmaryrd}
\usepackage{tikz-cd}
\usepackage{multirow,bigdelim}
\usepackage{array}
\usepackage[disable]{todonotes}
\usepackage{enumitem}
\usepackage{caption}
\usepackage{subcaption}
\usepackage{wrapfig}

\usepackage[noend]{algpseudocode}
\usepackage{algorithm}
\algrenewcommand\algorithmicindent{2ex}%

\usepackage{thmtools}
\usepackage{thm-restate}
\input{fix-restatable.tex} 
\usepackage{thm-patch}

\usepackage{url}
\usepackage[all,cmtip]{xy}
\usepackage{cleveref}
\usepackage{microtype}

\usetikzlibrary{cd,shapes,automata,arrows,fit,arrows.meta,positioning}

\makeatletter%
\spn@wtheorem{assumption}{Assumption}{\bfseries}{\itshape} 
\makeatother%
\newcommand{\qedhere}{\ensuremath{\qed}}

\mathcode`<="4268 
\mathcode`>="5269 
\mathchardef\lt="213C 
\mathchardef\gt="213E 

\newcommand{\cat}[1]{\mathbf{#1}}

\DeclareMathOperator{\id}{id}

\newcommand{\N}{\mathbb{N}}

\newcommand{\eword}{\varepsilon}

\usepackage{xspace}
\newcommand{\LStar}{\ensuremath{\mathtt{L}^{\!\star}}\xspace}

\newcommand{\close}{\mathsf{close}}
\newcommand{\cons}{\mathsf{cons}}
\newcommand{\lclose}{\mathsf{lclose}}
\newcommand{\lcons}{\mathsf{lcons}}

\newcommand{\aut}{\mathcal{A}}
\newcommand{\hyp}{\mathcal{H}}
\newcommand{\lang}{\mathcal{L}}

\newcommand{\epi}{\mathcal{E}}
\newcommand{\mono}{\mathcal{M}}

\newcommand{\wrap}{\mathcal{W}}

\newcommand{\arity}{\mathsf{arity}}
\newcommand{\pow}{\mathcal{P}}
\newcommand{\fpow}{\pow_{\mathsf{f}}}

\newcommand{\ot}{\mathsf{T}}

\newcommand{\rch}{\mathsf{reach}}

\newcommand*{\circled}[1]{{\tikz[baseline={(X.base)}]\node(X)[draw,shape=circle,inner sep=0]{\text{\scriptsize\strut$#1$}};}}

\tikzset{automaton/.style={node distance=1.5cm,on grid,auto,initial text={},state/.append style={inner sep=0pt,outer sep=0pt,minimum size=3.5ex}}}

\makeatletter
\def\orcidID#1{\href{http://orcid.org/#1}{\protect\raisebox{-1.25pt}{\protect\includegraphics{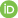}}}}
\makeatother

\newcommand{\matteo}[1]{\todo[size=\tiny,caption={},color=blue!30]{M:\@ #1}}
\newcommand{\tobias}[1]{\todo[size=\tiny,caption={},color=green!30]{T:\@ #1}}

\newcommand{\jurriaan}[1]{\todo[size=\tiny,caption={},color=purple!50]{J:\@ #1}}

\usepackage{etoolbox}
\makeatletter
\patchcmd{\@addmarginpar}{\ifodd\c@page}{\ifodd\c@page\@tempcnta\m@ne}{}{}
\makeatother
\reversemarginpar%

\input{columns.tex}

\tikzset{trlab/.style={font=\scriptsize,inner sep=1pt,outer sep=1pt}}

\begin{document}

%

\authorrunning{Van Heerdt et al.}

\title{A Categorical Framework for Learning Generalised Tree Automata}
\titlerunning{Learning Generalised Tree Automata}
\author{
Gerco~van Heerdt\inst{1}\orcidID{0000-0003-0669-6865} \and
Tobias~Kapp\'e\inst{2}\orcidID{0000-0002-6068-880X} \and
Jurriaan~Rot\inst{3} \and
Matteo~Sammartino\inst{4,1}\orcidID{0000-0003-1456-2242} \and
Alexandra~Silva\inst{5,1}\orcidID{0000-0001-5014-9784}}
\institute{%
    University College London \\ \email{gerco.heerdt@ucl.ac.uk} \and
    ILLC, University of Amsterdam \\ \email{t.kappe@uva.nl} \and
    Radboud University \\ \email{jrot@cs.ru.nl} \and
    Royal Holloway, University of London \\ \email{matteo.sammartino@rhul.ac.uk} \and
    Cornell University \\ \email{alexandra.silva@cornell.edu}
}

\maketitle

\begin{abstract}
	Automata learning is a popular technique used to automatically construct an automaton model from queries.
	Much research went into devising ad hoc adaptations of algorithms for different types of automata.
    The CALF project seeks to unify these using category theory in order to ease correctness proofs and guide the design of new algorithms.
	In this paper, we extend CALF to cover learning of algebraic structures that may not have a coalgebraic presentation.
	Furthermore, we provide a detailed algorithmic account of an abstract version of the popular \LStar algorithm, which was missing from CALF\@.
	We instantiate the abstract theory to a large class of $\cat{Set}$ functors, by which we recover for the first time practical tree automata learning algorithms from an abstract framework and at the same time obtain new algorithms to learn algebras of quotiented polynomial functors.
\end{abstract}

\section{Introduction}

Automata learning---automated discovery of automata models from system observations---is emerging as a highly effective bug-finding technique with applications in verification of passports~\cite{DBLP:conf/isola/AartsSV10}, bank cards~\cite{DBLP:conf/icst/AartsRP13}, and network protocols~\cite{DBLP:conf/cav/Fiterau-Brostean16}. The design of algorithms for automata learning of different models is a fundamental research problem, and in the last years much progress has been made in developing and understanding new algorithms. The roots of the area go back to the 50s, when Moore studied the problem of inferring deterministic finite automata. Later, the same problem, albeit under different names,  was studied by control theorists~\cite{gold} and computational linguists~\cite{higuera}. The algorithm that caught the attention of the verification community is the one presented in Dana Angluin's seminal paper in 1987~\cite{Angluin87}. She proves that it is possible to infer minimal deterministic automata in polynomial time using only so-called membership and equivalence queries. Vaandrager's CACM article~\cite{cacm-frits} provides an extensive review of the literature in automata learning and its applications to verification.

Angluin's algorithm, called \LStar, has served as a basis for many extensions that work for more expressive models than plain deterministic automata: I/O automata~\cite{DBLP:conf/concur/AartsV10}, weighted automata~\cite{Bergadano,fossacs20}, register automata~\cite{DBLP:conf/cav/IsbernerHS15,DBLP:journals/sigsoft/MuesHLKR16,AartsFKV15}, nominal automata~\cite{DBLP:conf/popl/MoermanS0KS17}, and B\"uchi automata~\cite{DBLP:journals/tcs/AngluinF16}. Many of these extensions were developed independently and, though they bear close resemblance to the original algorithm, arguments of correctness and termination had to be repeated every time. This motivated Silva and Jacobs to provide a categorical understanding of $\LStar$~\cite{jasi:auto14} and capture essential data structures abstractly, in the hope of developing a generic, modular, and parametric framework for automata learning based on (co)algebra. Their early work was taken much further in Van Heerdt's master thesis~\cite{heer:anab16}, which then formed the basis of a wider project on developing a \textit{Categorical Automata Learning Framework}---CALF.\footnote{\url{http://www.calf-project.org}}  CALF was described in the 2017 paper~\cite{HeerdtS017}, but several problems were left open: 
\begin{enumerate}[leftmargin=*]
\item An abstract treatment of counterexamples: in the original $\LStar$ algorithm, counterexamples are a core component, as they enable refinement of the state space of the learned automaton to ensure progress towards termination.
\item The development of a full abstract learning algorithm that could readily be instantiated for a given model: in essence, CALF provided only the abstract data structures needed in the learning process, but no direct algorithm.
\item Finding suitable constraints on the abstract framework to cover interesting examples, such as tree automata~\cite{hubert-etal-2008}, that did not fit the constraints in~\cite{HeerdtS017}.
\end{enumerate}
In this paper, we resolve the open problems above, and develop CALF further to provide concrete learning algorithms for models that are algebras for a given functor, which notably include tree automata. In a nutshell, the contributions and technical roadmap of the paper are as follows. After recalling some categorical notions, the basics of \LStar (\autoref{sec:prelim}), and CALF (\autoref{sec:calf}), we provide:
\begin{enumerate}[leftmargin=*]
\item A general treatment of counterexamples (\autoref{sec:counter}), together with an abstract analysis of progress, that enables termination analysis of a generic algorithm.
\item A step-by-step generalisation of all components of \LStar for models that are algebras of a given functor (\autoref{sec:alg}).
\item An instantiation of the abstract algorithm to concrete categories (\autoref{sec:examples}), providing the first abstractly derived learning algorithm for tree automata.
\end{enumerate}
The present paper complements other recent work on abstract automata learning algorithms: Barlocco, Kupke, and Rot~\cite{BarloccoKR19} gave an algorithm for coalgebras of a functor, whereas Urbat and Schr\"oder~\cite{abs-1911-00874} provided an algorithm for structures that can be represented as both algebras and coalgebras.
More recently, Colcombet, Petrisan, and Stabile~\cite{colcombet2021} gave an abstract learning algorithm based on modelling automata as functors.
Our focus is on algebras, such as tree automata, that cannot be covered by the aforementioned frameworks.
A detailed comparison is given in \autoref{sec:related}. We conclude with directions for future work in \autoref{sec:future}.

\section{Preliminaries}\label{sec:prelim}

We now introduce some categorical notions that we will need later in our technical development, and describe Angluin's original \LStar algorithm. We assume some prior knowledge of category theory (categories, functors); see e.g.,~\cite{awodey,lane1998categories}.

An \emph{$(\epi, \mono)$-factorisation system} on a category $\cat{C}$ consists of classes of morphisms $\epi$ and $\mono$, closed under composition with isos, such that for every morphism $f$ in $\cat{C}$ there exist $e \in \epi$ and $m \in \mono$ with $f = m \circ e$, and we have a unique diagonal fill-in property.
Given a morphism $f$, we write $f^\triangleright$ and $f^\triangleleft$ for the $\epi$-part and $\mono$-part of its factorisation, respectively.

We work in a category $\cat{C}$ with finite products and coproducts.
When $f: X \to Z$ and $g: Y \to Z$, we write $[f, g]$ for the unique arrow from $X + Y$ to $Z$ induced by the coproduct.
We assume that $\cat{C}$ admits a fixed factorisation system $(\epi, \mono)$, where $\epi$ consists of epis and $\mono$ consists of monos.
We fix a \emph{varietor} $F$ in $\cat{C}$, that is, an endofunctor such that there is a free $F$-algebra monad $(T, \eta, \mu)$.
We write $\gamma_X$ for the $F$-algebra structure $FTX \to TX$, which is natural in $X$.
Given an $F$-algebra $(Y, y)$, we write $f^\sharp \colon (TX, \mu_X) \to (Y, y)$ for the extension of $f \colon X \to Y$ and denote $y^* = \id_Y^\sharp \colon (TY, \mu_Y) \to (Y, y)$.
We often implicitly apply forgetful functors.
We fix an input object $I$ and an output object $O$ and write $F_I$ for the functor $I + F(-)$.
Lastly, we assume $F$ preserves $\epi$.

\subsection{Abstract Automata}

We recall the automaton definition from Arbib and Manes~\cite{arbib-manes-1974}, which we will use in this paper, and its basic properties of accepted language and minimality.

\begin{definition}[Automaton]%
	\label{def:automaton}
	An \emph{automaton} is a tuple $\aut = (Q, \delta, i, o)$ consisting of a state space object $Q$, dynamics $\delta \colon FQ \to Q$, initial states $i \colon I \to Q$, and an output $o \colon Q \to O$.
    A \emph{homomorphism} from $\aut$ to $\aut' = (Q', \delta', i', o')$ is an $F$-algebra morphism $h$ from $(Q, \delta)$ to $(Q', \delta')$ --- that is to say, a function $h\colon Q \to Q'$ with $\delta' \circ Fh = h \circ \delta$ --- such that $h \circ i = i'$ and $o' \circ h = o$.
\end{definition}
We will use the case of deterministic automata as a running example.
\begin{example}\label{ex:da}
	If $\cat{C} = \cat{Set}$ with the (surjective, injective) factorisation system, $F = (-) \times A$ for a finite set $A$, $I = 1 = \{*\}$, and $O = 2 = \{0, 1\}$, we recover deterministic automata (DAs) as automata: the state space is a set $Q$, the transition function is the dynamics, the initial state is represented as a function $1 \to Q$, and the classification of states into accepting and rejecting ones is represented by a function $Q \to 2$.
	In this case we obtain the monad $T = (-) \times A^*$, with its unit pairing an element with the empty word $\eword$ and the multiplication concatenating words.
	The extension of $\delta \colon Q \times A \to Q$ to $\delta^* \colon Q \times A^* \to Q$ is the usual one that lets the automaton read a word starting from a given state.
\end{example}
\begin{definition}[Language]
	A \emph{language} is a morphism $TI \to O$.
	The language accepted by an automaton $\aut = (Q, \delta, i, o)$ is given by $\lang_\aut = TI \xrightarrow{\rch_\aut} Q \xrightarrow{o} O$, where $\rch_\aut \colon TI \to Q$ is the \emph{reachability map} of $\aut$ given by $i^\sharp$.
\end{definition}

\begin{definition}[Minimality~\cite{arbib-manes-1974}]
	An automaton $\aut$ is said to be \emph{reachable} if $\rch_\aut \in \epi$.
	$\aut$ is \emph{minimal} if it is reachable and every reachable automaton $\aut'$ s.t.\ $\lang_\aut = \lang_{\aut'}$ admits a (necessarily unique) homomorphism to $\aut$.
\end{definition}

\begin{example}
	Recall the setting from \autoref{ex:da}.
	The reachability map $\rch_\aut \colon 1 \times A^* \to Q$ for a DA $\aut = (Q, \delta, i, o)$ assigns to each word the state reached after reading that word from the initial state.
	The language $\lang_\aut \colon 1 \times A^* \to 2$ accepted by $\aut$ is precisely the language accepted by $\aut$ in the traditional sense.
	Reachability of $\aut$ means that for every state $q \in Q$ there exists a word that leads to $q$ from the initial state.
	If this is the case, the unique homomorphism into a language-equivalent minimal automaton identifies states that accept the same language.
	Here, minimality is equivalent to having a minimal number of states.
\end{example}
A general study of existence of minimal automata in this setting is given in~\cite{AT89}; see also~\cite{HeerdtKRS019}.

\subsection{The \texorpdfstring{\LStar}{L-star}\ algorithm}\label{sec:lstar}
\begin{figure*}[t]
\begin{minipage}[t]{.66\textwidth}
\begin{algorithm}[H]
	\caption{Make table closed and consistent}\label{alg:lstar-closedcons}
	\begin{algorithmic}[1]
		\Function {Fix}{$S, E$}
			\While{$\ot$ is not closed or not consistent}
				\If{$\ot$ is not closed}
					\State find $t \in S,a \in A$ such that $\forall s \in S.\ \ot(ta) \neq \ot(s)$
					\State $S \gets S \cup \{sa\}$
				\ElsIf{$\ot$ is not consistent}
					\State find $s_1,s_2 \in S$, $a \in A$ and $e \in E$ such that
					\Statex \hspace{14mm} $\ot(s_1) = \ot(s_2)$ and $\ot(s_1a)(e) \neq \ot(s_2a)(e)$
					\State $E \gets E \cup \{ae\}$
				\EndIf
			\EndWhile
			\State \Return $S$, $E$
		\EndFunction
	\end{algorithmic}
\end{algorithm}
\end{minipage}
\quad
\begin{minipage}[t]{.30\textwidth}
\begin{algorithm}[H]
	\caption{\LStar algorithm}\label{alg:lstar-main}
	\begin{algorithmic}[1]
		\State $S \gets \{  \eword \}$\label{line:lstar-trivial-s}
		\State $E \gets \{  \eword \}$\label{line:lstar-trivial-e}
		\State $S, E \gets \Call{Fix}{S, E}$
		\While{$\mathsf{EQ}(\hyp_{\ot}) = c$}\label{line:lstar-loop}
			\State $S \gets S \cup \mathsf{prefixes}(c)$\label{line:lstar-update}
			\State $S, E \gets \Call{Fix}{S, E}$\label{line:lstar-complete}
		\EndWhile
		\State \Return $\hyp_{\ot}$\label{line:lstar-terminate}
	\end{algorithmic}
\end{algorithm}
\end{minipage}
\end{figure*}

In this section, we recall Angluin's algorithm \LStar, which learns the \emph{minimal} DFA accepting a given unknown regular language $\lang$. The algorithm can be seen as a game between two players: a \emph{learner} and a \emph{teacher}. The learner can ask two types of \emph{queries} to the teacher:
\begin{enumerate}[leftmargin=*]
	\item
		\textbf{Membership queries}:  is a word $w \in A^*$ in $\lang$?
	\item
		\textbf{Equivalence queries}: is a \emph{hypothesis} DFA $\hyp$ correct? That is, is $\lang_\hyp = \lang$?\end{enumerate}

The teacher answers \emph{yes} or \emph{no} to these queries. Moreover, negative answers to equivalence queries are witnessed by a \emph{counterexample}---a word classified incorrectly by $\hyp$.
The learner gathers the results of queries into an \emph{observation table}: a function $\ot \colon S \cup S \cdot A \to 2^E$, where $S, E\subseteq A^*$ are finite and $\ot(s)(e) = \lang(se)$.  This function can be depicted as a table where elements of $S \cup S\cdot A$ label rows ($\cdot$ is pointwise concatenation) and elements of $E$ label columns.

\begin{wrapfigure}{r}{4cm}
\vspace{-.5cm}
\begin{center}
\begin{tabular}{L{4em} R{2ex} | C{2.5ex}C{2.5ex}C{2.5ex}@{}m{0pt}@{}} 
        \multicolumn{6}{c}{\hfill$\overbracket[.8pt][2pt]{\rule{8ex}{0pt}}^{\displaystyle E}$} \\
		& & \eword &  b & ab &
		\\
		\cline{2-6}
		\ldelim\lbrack{1}{5mm}[$S$\,] & \eword & 1 & 0 & 1 &
		\\[1ex]
		\cline{2-6}
		\ldelim\lbrack{2}{3em}[$S \cdot A$\,]
		& a & 0 & 1 & 0 &
		\\
		& b & 0 & 0 & 0 &
\end{tabular}
\end{center}
\vspace{-.5cm}
\end{wrapfigure}
As an example, consider the table on the right, over the alphabet $A = \{a,b\}$, where $S = \{ \eword \}$ and $E = \{ \eword, b, ab \}$.
This table approximates a language that contains $\eword, ab$, but not $a,b,bb,aab,bab$. Following the visual intuition, we will refer to the part of the table indexed by $S$ as the \emph{top} part of the table, and the one indexed by $S \cdot A$ as the \emph{bottom} part.

Intuitively, the content of each row labelled by a word $s$ approximates the Myhill--Nerode equivalence class of $s$. This is in fact the main idea behind the construction of a hypothesis DFA $\hyp_\ot$ from $\ot$: states of $\hyp_\ot$ are distinct rows of $\ot$, corresponding to distinct Myhill--Nerode equivalence classes. Formally, $\hyp_{\ot} = (Q,q_0,\delta,F)$ is defined as follows: 
\begin{itemize}[leftmargin=*]
\item $Q = \{\ot(s)\mid s\in S\}$ is the set of states;
\item $F = \{ \ot(s) \mid s\in S, \ot(s)( \eword)=1\}$ is the set of final states;
\item  $q_0 = \ot(\eword)$ is the initial state;
\item $\delta\colon Q\times A \to Q, (\ot(s),a) \mapsto \ot(sa)$ is the transition function.
\end{itemize}
For $F$ and $q_0$ to be well-defined we need $\eword$ in $E$ and $S$ respectively. Moreover, for $\delta$ to be well-defined we need $\ot(sa) \in Q$ for all $sa \in S \cdot A$, and we must ensure that the choice of $s$ to represent a row does not affect the transition. These constraints are captured in the following two properties.
\begin{definition}[Closedness and consistency]%
\label{def:lstar-clos-cons}
A table $\ot$ is \emph{closed} if for all $t \in S$ and $a \in A$ there exists $s \in S$ such that $\ot(s) = \ot(ta)$. A table is \emph{consistent} if for all $s_1, s_2 \in S$ with $\ot(s_1) = \ot(s_2)$ we have $\ot(s_1a) = \ot(s_2a)$ for any $a \in A$.
\end{definition}
Closedness and consistency form the core of \LStar, described in \autoref{alg:lstar-main}. The sets $S$ and $E$ are initialised with the empty word $\eword$ (lines~\ref{line:lstar-trivial-s} and~\ref{line:lstar-trivial-e}), and extended as a closed and consistent table is built using the subroutine \textsc{Fix}, given in \autoref{alg:lstar-closedcons}. The main loop uses an equivalence query, denoted $\mathsf{EQ}$, to ask the teacher whether the hypothesis induced by the table is correct. If the result is a counterexample $c$, the table is updated by adding all prefixes of $c$ to $S$ (line~\ref{line:lstar-update}) and made closed and consistent again (line~\ref{line:lstar-complete}).
Otherwise, the algorithm returns with the correct hypothesis (line~\ref{line:lstar-terminate}). See Appendix~\ref{app:lstar-ex} for an example.

\section{The abstract data structures in CALF}\label{sec:calf}

We recall the basic notions underpinning CALF~\cite{HeerdtS017}: generalisations of the observation table, closedness, consistency and hypothesis. The generalised table is called a \emph{wrapper}:

\begin{definition}[Wrapper]%
	\label{def:wrapper}
    A \emph{wrapper} for an object $Q$ is a pair of morphisms
    \[\wrap = \big(\hspace{-1mm}\xymatrix@C=.4cm{S\ar[r]^{ \;\alpha\;}& Q}, \xymatrix@C=.4cm{Q\ar[r]^{ \;\beta\;}& P}\hspace{-1mm}\big)\]
    We denote the factorisation of $\beta \circ \alpha$ by $\xymatrix@C=.8cm{S  \ar@{->>}[r]^(.4){\ e_\wrap \ }& H_\wrap \ar@{>->}[r]^(.6){\ m_\wrap \ } & P}$.
\end{definition}
This will be instantiated with $Q$ the state space of the target automaton, $S$ a collection of
row labels of an observation table, and $P$ a collection of possible values of the rows. Then
$\alpha$ \emph{selects} states in $Q$, and $\beta$ \emph{classifies} them into $P$.
We note that although such $\alpha$ and $\beta$ underly the learning algorithm, they are not actually known to the learner, as they explicitly involve the unknown target automaton.
However, we will see that we only need to represent certain compositions involving these morphisms, and that when $\alpha$ and $\beta$ are chosen appropriately it will be possible to compute these compositions.

\begin{example}[Observation table wrapper]\label{ex:dawrapper}
	Recall the DA setting from \autoref{ex:da} and consider a DA $\aut = (Q, \delta, i, o)$.
    For $S \subseteq A^*$ and $E \subseteq A^*$, we can define a wrapper $\wrap = \big(\xymatrix@C=.6cm{S\ar[r]^(.45){ \;\alpha_S\;}& Q}, \xymatrix@C=.6cm{Q\ar[r]^(.55){ \;\beta_E\;}& 2^E}\big)$ for $Q$ as follows:
    \begin{mathpar}
        \alpha_S(w) = \rch_\aut(*, w)
        \and
        \beta_E(q)(e) = (o \circ \delta^*)(q, e).
	\end{mathpar}
	The composition $\beta_E \circ \alpha_S \colon S \to 2^E$ is precisely the top part of the observation table of \LStar, with rows $S$ and columns $E$. In fact, we have
	$
		(\beta_E \circ \alpha_S)(s)(e) = \lang_\aut(*, se).
	$
	The image of $\beta_E \circ \alpha_S$ is the set of rows that appear in the table. In \LStar{}, this set is used as states of the hypothesis, and in our setting can be obtained
    as $H_\wrap$, recalling that the (surjective, injective) factorisation system in $\cat{Set} $ gives factorisation through the image.
\end{example}
Before we define hypotheses in this abstract framework, we need generalised notions of closedness and consistency.
\begin{definition}[Closedness and consistency]%
	\label{def:clos-cons}
	Given a wrapper
    \[\wrap = \big(\hspace{-1mm}\xymatrix@C=.4cm{S\ar[r]^{ \;\alpha\;}& Q}, \xymatrix@C=.4cm{Q\ar[r]^{ \;\beta\;}& P}\hspace{-1mm}\big),\]
    where $Q$ is the state space of an automaton $(Q, \delta, i, o)$, we say that $\wrap$ is \emph{closed} if there exist morphisms $i_\wrap \colon I \to H_\wrap$ and $\close_\wrap \colon FS \to H_\wrap$ making the diagrams below commute.
	\begin{align*}
		\begin{tikzcd}[column sep=1.2cm,ampersand replacement=\&]
			I \ar{r}{i} \ar[dashed]{d}[swap]{i_\wrap} \&
				Q \ar{d}{\beta} \\
			H_\wrap \ar{r}{m_\wrap} \&
				P
		\end{tikzcd}
        &&
        \begin{tikzcd}[column sep=.7cm,ampersand replacement=\&]
            FS \ar{r}{F\alpha} \ar[dashed]{d}[swap]{\close_\wrap} \&
                FQ \ar{r}{\delta} \&
                Q \ar{d}{\beta} \\
            H_\wrap \ar{rr}{m_\wrap} \&
                \&
                P
        \end{tikzcd}
	\end{align*}
	Furthermore, we say that $\wrap$ is \emph{consistent} if there exist morphisms $o_\wrap \colon H_\wrap \to O$ and $\cons_\wrap \colon FH_\wrap \to P$ making the diagrams below commute.
	\begin{align*}
		\begin{tikzcd}[column sep=1.2cm,ampersand replacement=\&]
			S \ar{r}{e_\wrap} \ar{d}[swap]{\alpha} \&
				H_\wrap \ar[dashed]{d}{o_\wrap} \\
			Q \ar{r}{o} \&
				O
		\end{tikzcd}
        &&
        \begin{tikzcd}[column sep=.7cm,ampersand replacement=\&]
            FS \ar{d}[swap]{F\alpha} \ar{rr}{Fe_\wrap} \&
                \&
                FH_\wrap \ar[dashed]{d}{\cons_\wrap} \\
            FQ \ar{r}{\delta} \&
                Q \ar{r}{\beta} \&
                P
        \end{tikzcd}
	\end{align*}
\end{definition}

\begin{example}%
  	\label{ex:wrap-clos-cons}
	In the DA case, generalised closedness and consistency instantiate to the conditions allowing the hypothesis to be well-defined in \LStar{} (see \autoref{sec:lstar}):

	\begin{description}[noitemsep,topsep=0pt]
        \item[Closedness:]	The wrapper $(\alpha_S, \beta_E)$ is closed if: (i)~there exists $s \in S$ such that $(\beta_E \circ \alpha_S)(s) = (\beta_E \circ i)(*)$ and; (ii)~for all $s \in S$ and $a \in A$ there exists $s_a \in S$ such that $(\beta_E \circ \alpha_S)(s_a) = (\beta_E \circ \delta)(\alpha_S(s), a)$. Condition (i) holds immediately if $\eword \in S$---the function $(\beta_E \circ i)(*) \colon E \to 2$ maps $e \in E$ to $\lang_\aut(*, e)$. Condition (ii) corresponds to closedness in \autoref{def:lstar-clos-cons}. In fact, $\beta_E \circ \delta \circ (\alpha_S \times \id_A) \colon S \times A \to 2^E$ represents the lower part of the observation table associated with $S$ and $E$.

	\item[Consistency:]
    The wrapper $(\alpha_S, \beta_E)$ is consistent if: (iii)~for all $s_1, s_2 \in S$ such that $(\beta_E \circ \alpha_S)(s_1) = (\beta_E \circ \alpha_S)(s_2)$ we have $(o \circ \alpha_S)(s_1) = (o \circ \alpha_S)(s_2)$ and; (iv)~for all $a \in A$ we have $(\beta_E \circ \delta)(\alpha_S(s_1), a) = (\beta_E \circ \delta)(\alpha_S(s_2), a)$.
            Condition (iii) holds immediately if $\eword \in E$---the function $o \circ \alpha_S \colon S \to 2$ maps $s \in S$ to $\lang_\aut(*, s)$. Condition (iv) corresponds to consistency in \autoref{def:lstar-clos-cons}.
  	\end{description}
    To determine these properties, we do not need the individual descriptions of $\alpha_S$ and $\beta_E$, which refer to the target automaton and are thus not available to the learner; we just need the compositions $\beta_E \circ \alpha_S$, $\beta_E \circ i$, $\beta_E \circ \delta \circ (\alpha_S \times \id_A)$, and $o \circ \alpha_S$, which can be determined using membership queries in this case.
    In general, for any instantiation of our abstract algorithm it will be important to show that these compositions (adapted to the wrapper and functor involved) can be determined and used concretely by the instantiated algorithm.
\end{example}
So far, we have used the wrapper to obtain the state space $\hyp_\wrap$ of the hypothesis. When a wrapper is closed and consistent, we can equip $\hyp_\wrap$ with a full automaton structure, leveraging the unique diagonal fill-in property of the factorisation.
\begin{definition}[Hypothesis]%
	\label{def:hypothesis}
	A closed and consistent wrapper \[\wrap = \big(\hspace{-1mm}\xymatrix@C=.4cm{S\ar[r]^{ \;\alpha\;}& Q}\hspace{-0.5mm},\hspace{-0.5mm}\xymatrix@C=.4cm{Q\ar[r]^{ \;\beta\;}& P}\hspace{-1mm}\big)\] for $(Q, \delta, i, o)$ induces a \emph{hypothesis} automaton $\hyp_\wrap = (H_\wrap, \delta_\wrap, i_\wrap, o_\wrap)$, where $\delta_\wrap$ is the unique diagonal in the commutative square below.
	\[
		\begin{tikzcd}
			FS \ar[twoheadrightarrow]{r}{Fe_\wrap} \ar{d}[swap]{\close_\wrap} &
				FH_\wrap \ar{d}{\cons_\wrap} \ar[dashed]{dl}[swap]{\delta_\wrap} \\
			H_\wrap \ar[rightarrowtail]{r}{m_\wrap} &
				P
		\end{tikzcd}
	\]
\end{definition}

\section{Counterexamples, generalised}\label{sec:counter}
We now provide a key missing element for the development and analysis of an abstract learning algorithm in CALF:\@ counterexamples. In the original \LStar algorithm, counterexamples are used to refine the state space of the hypothesis---namely the representations of the Myhill--Nerode classes of the language being learned. A crucial property for termination, which we prove at a high level of generality in this section, is that adding counterexamples to a closed and consistent table results in a table which is either not closed or not consistent, and hence needs to be extended. Such an extension, in turn, results in progress being made in the algorithm. We show how we can use \emph{recursive coalgebras}~\cite{osius1974,taylor1999} as witnesses for discrepancies---i.e., as counterexamples---between a hypothesis and the target language in our abstract approach.\footnote{Recursive coalgebras have been used to generalise prefix-closedness in an automata learning context in earlier work~\cite{vanheerdt2018}, as well as to generalise counterexamples~\cite{BarloccoKR19,abs-1911-00874}.} Here, and throughout the paper, we fix a \emph{target} automaton $\aut_\mathsf{t} = (Q_\mathsf{t}, \delta_\mathsf{t}, i_\mathsf{t}, o_\mathsf{t})$ whose language we want to learn. 

\begin{definition}[Recursive coalgebras]
	An $F$-coalgebra $\rho \colon S \to F S$ is \emph{recursive} if for every algebra $x \colon F X \to X$ there is a unique morphism $x^\rho \colon S \to X$ making the diagram below commute.
	\[
		\begin{tikzcd}[column sep=1.1cm]
			F S \ar[dashed]{r}{F x^\rho} &
				F X \ar{d}{x} \\
			S \ar[dashed]{r}{x^\rho} \ar{u}{\rho} &
				X
		\end{tikzcd}
	\]
\end{definition}

\begin{example}%
	\label{ex:pref-closed-rec}
	A \emph{prefix-closed} subset $S \subseteq A^*$ is easily equipped with a coalgebra structure $\rho \colon S \to 1 + S \times A$ that detaches the last letter from each non-empty word and assigns $*$ to the empty one.
	Such a coalgebra is recursive, with the unique map into an algebra being defined as a restricted reachability map.
	In fact, under certain conditions that are satisfied in the DA setting, recursivity of a coalgebra is equivalent to having a coalgebra homomorphism into the initial algebra~\cite[Corollary 5.6]{AdamekMM20}.
	This means that every recursive coalgebra is isomorphic to one given by a prefix-closed multiset of words.
	If the unique morphism into the initial algebra is injective, then the multiset becomes a set.
\end{example}
Given an automaton $\aut = (Q,\delta,i,o)$ and a recursive coalgebra $\rho \colon S \to F_I S$, the map $S \xrightarrow{{[i, \delta]}^\rho} Q$ can be seen as a generalised reachability map, allowing states in $Q$ to be reached from $S$. We use this map to derive a notion of generalised language induced by a recursive coalgebra.
This will be used to compare languages of the hypothesis and of the target automaton with respect to a specific recursive coalgebra, i.e., a specific counterexample.
\begin{definition}[$\rho$-languages]
	Given a recursive coalgebra $\rho \colon S \to F_I S$ and an automaton $\aut = (Q, \delta, i, o)$, the $\rho$-language of $\aut$ is
	$
		\lang_\aut^\rho = S \xrightarrow{{[i, \delta]}^\rho} Q \xrightarrow{o} O
	$.
\end{definition}
For instance, in the case of a DA $\aut$ and a recursive coalgebra as in \autoref{ex:pref-closed-rec}, $\lang_\aut^\rho$ is the restriction of the language of $\aut$ to the prefix-closed set of words $S$.

In \autoref{alg:lstar-main}, a counterexample is produced by the teacher (line~\ref{line:lstar-loop}) when the hypothesis does not agree with the target automaton. We now generalise counterexamples to wrappers: counterexamples are recursive coalgebras on which the languages of the hypothesis and of the target automaton disagree.
\begin{definition}[Counterexample]
	A closed and consistent wrapper $\wrap$ is said to be \emph{correct up to} a recursive $\rho \colon S \to F_I S$ if $\lang_{\hyp_\wrap}^\rho = \lang_{\aut_\mathsf{t}}^\rho$.
	A \emph{counterexample} for $\wrap$ (or $\hyp_\wrap$) is a recursive $\rho \colon S \to F_I S$ such that $\wrap$ is not correct up to $\rho$.
\end{definition}

The following guarantees incorrect hypotheses yield counterexamples.
\begin{restatable}[Language equivalence via recursion]{proposition}{sublang}\label{prop:sublang}
	Given an automaton $\aut = (Q, \delta, i, o)$, we have $\lang_{\aut_\mathsf{t}} = \lang_{\aut}$ if and only if $\lang_{\aut_\mathsf{t}}^\rho = \lang_{\aut}^\rho$ for every  recursive coalgebra $\rho \colon S \to F_I S$.
\end{restatable}

\begin{corollary}[Counterexample existence]
	Given a closed and consistent wrapper $\wrap$ for $Q_\mathsf{t}$, we have $\lang_{\hyp_\wrap} \ne \lang_{\aut_\mathsf{t}}$ iff there exists a counterexample for $\wrap$.
\end{corollary}

The next step in \autoref{alg:lstar-main} is to fix the table by adding all prefixes of the counterexample to $S$ (line~\ref{line:lstar-update}). We generalise this step by incorporating the counterexample given by a recursive coalgebra into the wrapper. In the DA case, this precisely corresponds to adding a prefix-closed subset to $S$. The following results say that doing so will lead to either a closedness or a consistency defect. In other words, we give theoretical guarantees that resolving counterexamples results in progress being made towards convergence.

\begin{restatable}[Resolving counterexamples]{theorem}{rescount}%
	\label{thm:rescount}
	Given a closed and consistent wrapper $\wrap = \big(\xymatrix@C=.4cm{S\ar[r]^{ \;\alpha\;}& Q_\mathsf{t}}, \xymatrix@C=.4cm{Q_\mathsf{t}\ar[r]^{ \;\beta\;}& P}\big)$ and a recursive coalgebra $\rho\colon S' \to F_I S'$, the following holds. If the wrapper $\wrap' = ([\alpha, {[i_\mathsf{t}, \delta_\mathsf{t}]}^\rho], \beta)$ is closed and consistent, then $\wrap$ is correct up to $\rho$.
\end{restatable}
This theorem is used contrapositively: given a closed and consistent wrapper, adding a counterexample yields a wrapper that is either not closed or inconsistent.

\section{Generalised Learning Algorithm}\label{sec:alg}

\begin{figure*}[t]
\begin{minipage}[t]{0.47\textwidth}
\begin{algorithm}[H]
	\caption{Make wrapper closed and consistent}\label{alg:closedcons}
	\begin{algorithmic}[1]
		\Function {Fix}{$\alpha, \beta$}
			\While{$(\alpha, \beta)$ not closed \Statex \qquad \qquad or not consistent}
				\If{$(\alpha, \beta)$ not closed}
					\State $\alpha \gets \alpha'$ such that $(\alpha', \beta)$ is
					\label{line:closed}
					\Statex \qquad \qquad locally closed w.r.t. $\alpha$
				\ElsIf{$(\alpha, \beta)$ not consistent}
					\State $\beta \gets \beta'$ such that $(\alpha, \beta')$ is
					\label{line:consist}
					\Statex \qquad\qquad locally consistent w.r.t.~$\beta$
				\EndIf
			\EndWhile
			\State \Return $\alpha, \beta$
		\EndFunction
	\end{algorithmic}
\end{algorithm}
\end{minipage}
\quad
\begin{minipage}[t]{0.48\textwidth}
\begin{algorithm}[H]
	\caption{Abstract automata learning algorithm}\label{alg:main}
	\begin{algorithmic}[1]
		\State $\alpha, \beta \gets \Call{Fix}{{!} \colon 0 \to Q_\mathsf{t}, {!} \colon Q_\mathsf{t} \to 1}$\label{line:trivial}
		\While{$\mathsf{EQ}(\hyp_{(\alpha, \beta)}) = \rho \colon S \to F_I S$}\label{line:loop}
			\State $\alpha \gets \alpha'$ s.t.\ $\alpha'^\triangleleft = [\alpha, [i_\mathsf{t}, \delta_\mathsf{t}]^\rho]^\triangleleft$
			\label{line:update}
			\State $\alpha, \beta \gets \Call{Fix}{\alpha, \beta}$\label{line:complete}
		\EndWhile
		\State \Return $\hyp_{(\alpha, \beta)}$\label{line:terminate}
	\end{algorithmic}
\end{algorithm}
\end{minipage}
\end{figure*}

We are now in a position to describe our general algorithm.
Similarly to \LStar~(\autoref{sec:lstar}) it is organised into two procedures: \autoref{alg:closedcons}, which contains the abstract procedure for making a wrapper closed and consistent, and \autoref{alg:main}, containing the learning iterations.
These generalise the analogous procedures in \LStar, \autoref{alg:lstar-closedcons} and \autoref{alg:lstar-main}, respectively.
We note again that although the algorithmic description operates on a wrapper $(\alpha, \beta)$, these individual morphisms will not be known to the learner.
In fact, at this level of abstraction the descriptions should be seen as algorithmic templates rather than concrete algorithms.
An instantiation must ensure that at least the compositions required to determine the closedness and consistency conditions and to construct the hypothesis can be maintained.
These compositions are $\beta \circ \alpha$, $\beta \circ i_{\mathsf{t}}$, $\beta \circ \delta_{\mathsf{t}} \circ F\alpha$, and $o_\mathsf{t} \circ \alpha$.
We have previously shown how these instantiate to recover \LStar, and in \autoref{sec:examples} we will discuss the class of examples given by generalised tree automata.

In \autoref{alg:main}, the wrapper is initialised with trivial maps and extended to be closed and consistent using the subroutine \textsc{Fix} (line~\ref{line:trivial}). The equivalence query for the main loop (line~\ref{line:loop}) returns a counterexample in the form of a recursive coalgebra, which is used to update the wrapper (line~\ref{line:update}, which will be explained in more detail later, when we define runs).
The updated wrapper is passed on to the subroutine \textsc{Fix} (line~\ref{line:complete}) to be made closed and consistent.

A crucial point for \autoref{alg:closedcons} is defining what it means to resolve the ``current'' closedness and consistency defects.
We call these \emph{local} defects, meaning the ones that can be directly detected in the current wrapper.
For DAs, local closedness defects are rows from the bottom part missing in the top part, and the empty word row if it is missing.
Local consistency defects are pairs of row labels which are distinguished by the target language, or with differing rows when the labels are extended with a single symbol.

We first introduce additional notions to formalise these ideas.
\jurriaan{should we define quotients and subobjects?}
\tobias{It's kind of implicit in the definition of ordering, is it not?}
We partially order the subobjects and quotients of the target automaton's state space $Q_\mathsf{t}$ in the usual way. Given two subobjects $j \colon J \hookrightarrow Q_\mathsf{t}$ and $k \colon K \hookrightarrow Q_\mathsf{t}$, we say $j \leq k$ if there is $f \colon J \to K$ such that $k \circ f = j$. Intuitively, $j$ is ``contained'' in $k$. Given two quotients $x \colon Q_\mathsf{t} \twoheadrightarrow X$ and $y \colon Q_\mathsf{t} \twoheadrightarrow Y$, we say $x \le y$ if there exists $g \colon X \to Y$ such that $y = g \circ x$. Intuitively, $x$ is ``finer'' than $y$.

Now, consider a wrapper $(\alpha,\beta)$ for $Q_\mathsf{t}$. We have that $\alpha^\triangleleft$ and $\beta^\triangleright$ are a subobject and a quotient of $Q_\mathsf{t}$, respectively. For instance, in the DA case, $\alpha^\triangleleft$ is the set of states in $Q_\mathsf{t}$ currently represented by the table, and $\beta^\triangleright$ is the equivalence relation on states induced by the rows. We can now say another wrapper $(\alpha',\beta')$ is a locally closed extension of $(\alpha,\beta)$ if (a)~it represents at least the same states of the target automaton as $\alpha$, formalised $\alpha^\triangleleft \le \alpha'^\triangleleft$, and; (b)~it solves the closedness defects present in $\alpha$.
Local consistency is analogous: it requires the extended wrapper to distinguish at least the same states of $Q_\mathsf{t}$ as the original one.

\begin{definition}[Local closedness and consistency]
	Consider a wrapper
    \[\wrap = \big(\xymatrix@C=.4cm{S'\ar[r]^{ \;\alpha'\;}& Q_\mathsf{t}}, \xymatrix@C=.4cm{Q_\mathsf{t}\ar[r]^{ \;\beta'\;}& P'}\big)\]
	We call $\wrap$ \emph{locally closed} w.r.t.\ a morphism $\alpha \colon S \to Q_\mathsf{t}$ if $\alpha^\triangleleft \le \alpha'^\triangleleft$ and there are morphisms $i_\wrap \colon I \to H_\wrap$ and $\lclose_{\wrap,\alpha} \colon FS \to H_\wrap$ s.t.\ these diagrams commute:
	\begin{align*}
		\begin{tikzcd}[column sep=1.2cm,ampersand replacement=\&,row sep=1.2em]
			I \ar{r}{i_\mathsf{t}} \ar[dashed]{d}[swap]{i_\wrap} \&
				Q_\mathsf{t} \ar{d}{\beta'} \\
			H_\wrap \ar{r}{m_\wrap} \&
				P'
		\end{tikzcd}
        &&
        \begin{tikzcd}[column sep=.7cm,ampersand replacement=\&,row sep=1.2em]
            FS \ar{r}{F\alpha} \ar[dashed]{d}[swap]{\lclose_{\wrap,\alpha}} \&
                FQ_\mathsf{t} \ar{r}{\delta_\mathsf{t}} \&
                Q_\mathsf{t} \ar{d}{\beta'} \\
            H_\wrap \ar{rr}{m_\wrap} \&
                \&
                P'
        \end{tikzcd}
	\end{align*}
	Given $\beta \colon Q \to P$, we say that $\wrap$ is \emph{locally consistent} w.r.t.\ $\beta$ if $\beta'^\triangleright \le \beta^\triangleright$ and there exist morphisms $o_\wrap \colon H_\wrap \to O$ and $\lcons_{\wrap,\beta} \colon FH_\wrap \to P$ making the diagrams below commute.
	\begin{align*}
		\begin{tikzcd}[column sep=1.2cm,ampersand replacement=\&,row sep=1.2em]
			S' \ar{r}{e_\wrap} \ar{d}[swap]{\alpha'} \&
				H_\wrap \ar[dashed]{d}{o_\wrap} \\
			Q_\mathsf{t} \ar{r}{o_\mathsf{t}} \&
				O
		\end{tikzcd}
        &&
        \begin{tikzcd}[column sep=.7cm,ampersand replacement=\&,row sep=1.2em]
            FS' \ar{d}[swap]{F\alpha'} \ar{rr}{Fe_\wrap} \&
                \&
                FH_\wrap \ar[dashed]{d}{\lcons_{\wrap,\beta}} \\
            FQ_\mathsf{t} \ar{r}{\delta_\mathsf{t}} \&
                Q_\mathsf{t} \ar{r}{\beta} \&
                P
        \end{tikzcd}
	\end{align*}
\end{definition}
A wrapper $(\alpha, \beta)$ is closed if and only if it is locally closed w.r.t.\ $\alpha$ and consistent if and only if it is locally consistent w.r.t.\ $\beta$.
\begin{example}%
    \label{ex:loc-clos-cons}
	For the case of DAs, consider a wrapper $(\alpha_{S'},\beta_{E'})$ representing an observation table $(S',E')$ for the target DA $\aut_\mathsf{t}$ (see \autoref{ex:dawrapper}):
	\begin{description}[noitemsep,topsep=0pt]
		\item[Local closedness:] Given $\alpha_S \colon S \to Q_\mathsf{t}$, $(\alpha_{S'},\beta_{E'})$ is locally closed w.r.t.\ $\alpha_S$ if (1)~$S \subseteq S'$ (to ensure $\alpha_S^\triangleleft \leq \alpha_{S'}^\triangleleft$); (2)~$S'$ contains the empty word (left diagram); and (3)~any row in the bottom part of the table $(S, E')$ occurs in the top part of $(S', E')$ (right diagram).
		\item[Local consistency:] Similarly, given $\beta_E \colon Q_\mathsf{t} \to 2^E$, $(\alpha_{S'},\beta_{E'})$ is locally consistent w.r.t.\ $\beta_E$ if (1)~$E \subseteq E'$ (to ensure $\beta_{E'}^\triangleright \le \beta_E^\triangleright$); (2)~$E'$ contains the empty word (left diagram); and (3)~for all $s, s' \in S'$ and $a \in A$, if $s$ and $s'$ map to the same row in the top part of $(S', E')$, then the rows for $sa$ and $s'a$ are the same in the bottom part of $(S', E)$ (right diagram).
	\end{description}
\end{example}
In \autoref{alg:closedcons} we assume that we can always find locally closed and consistent wrappers (lines~\ref{line:closed} and~\ref{line:consist} respectively). This assumption holds in general for local closedness: for each wrapper $(\alpha, \beta)$ for $Q_\mathsf{t}$ we can always find $\alpha'$ such that $(\alpha', \beta)$ is locally closed w.r.t.\ $\alpha$.

\begin{restatable}{lemma}{locallyclosed}
	Given a wrapper $(\alpha, \beta)$ for $Q_\mathsf{t}$,
	$([\alpha, [i_\mathsf{t}, \delta_\mathsf{t}] \circ F_I \alpha], \beta)$ is locally closed w.r.t.\ $\alpha$.
\end{restatable}
\noindent
This result is enabled by the algebraic
nature of automata. Local consistency is not inherently algebraic, so ensuring it takes more effort.
We shall see in \autoref{sec:examples} that existence of locally closed/consistent extensions can be proved constructively for a broad class of automata.

\paragraph*{Termination.}

To analyse termination of \autoref{alg:main}, we introduce its \emph{runs}.
\begin{definition}[Run of the algorithm]%
	\label{def:run}
	A \emph{run} of the algorithm is a stream of wrappers $\wrap_n = (\alpha_n, \beta_n)$ satisfying the following conditions:
	\begin{enumerate}[leftmargin=*]
		\item
			$\alpha_0 \colon 0 \to Q_\mathsf{t}$ and $\beta_0 \colon Q_\mathsf{t} \to 1$ are the unique morphisms;
		\item
			if $\wrap_n$ is not closed, then $\beta_{n + 1} = \beta_n$ and $\alpha_{n + 1}$ is s.t.\ $(\alpha_{n + 1}, \beta_n)$ is locally closed w.r.t.\ $\alpha_n$;
		\item
			if $\wrap_n$ is closed but not consistent, then $\alpha_{n + 1} = \alpha_n$ and $\beta_{n + 1}$ is s.t.\ $(\alpha_n, \beta_{n + 1})$ is locally consistent w.r.t.\ $\beta_n$;
		\item\label{item:cex}
			if $\wrap_n$ is closed and consistent and we obtain a counterexample $\rho \colon S \to F_I S$ for $\wrap_n$, then $\alpha_{n + 1}^\triangleleft = {[\alpha_n, {[i_\mathsf{t}, \delta_\mathsf{t}]}^\rho]}^\triangleleft$ and $\beta_{n + 1} = \beta_n$; and
		\item
			if $\wrap_n$ is closed and consistent and correct up to all recursive $F_I$-coalgebras, then $\wrap_{n + 1} = \wrap_n$.
	\end{enumerate}
\end{definition}
Note that, in point~\ref{item:cex} above and in line~\ref{line:update} of \autoref{alg:main}, we admit a more general counterexample resolution than \autoref{thm:rescount}: we only require that $\alpha_{n+1}$ and $[\alpha_n, {[i_\mathsf{t}, \delta_\mathsf{t}]}^\rho]$ represent the same states of the target automaton. This captures how observation tables are updated in practice; for instance in \LStar{} a counterexample prefix already in the table is discarded.%
\begin{proposition}\label{prop:converge}
	\autoref{alg:main} halts if and only if for all runs ${\{\wrap_n\}}_{n \in \N}$ there is $n$ with $\wrap_{n + 1} = \wrap_n$.
\end{proposition}
We can establish an invariant on the order of subsequent wrappers in runs.
\begin{restatable}{lemma}{alphabeta}\label{lem:alphabeta}
	Let ${\{\wrap_n = (\alpha_n, \beta_n)\}}_{n \in \N}$ be a run.
	For all $n \in \N$, we have $\alpha_n^\triangleleft \le \alpha_{n + 1}^\triangleleft$ and $\beta_{n + 1}^\triangleright \le \beta_n^\triangleright$.
	Moreover, if $\alpha_{n + 1}^\triangleleft \le \alpha_n^\triangleleft$, then $\alpha_{n + 1} = \alpha_n$; if $\beta_n^\triangleright \le \beta_{n + 1}^\triangleright$, then $\beta_{n + 1} = \beta_n$.
\end{restatable}
Putting these results together, we conclude that the algorithm terminates with a correct automaton, which is minimal under certain conditions. Satisfaction of the requirement on recursive coalgebras $\rho_k$ depends on the implementation of counterexamples and closing of wrappers; for DAs, it suffices to keep the set of row labels $S$ prefix-closed.

\begin{restatable}[Termination]{theorem}{termination}%
	\label{thm:termination}
	If $Q_\mathsf{t}$ has finitely many subobject and quotient isomorphism classes, then for all runs ${\{\wrap_n = (\alpha_n, \beta_n)\}}_{n \in \N}$ there exists $n \in \N$ such that $\wrap_n$ is closed and consistent and its hypothesis is correct.
	If $\aut_\mathsf{t}$ is minimal and for all $k \in \N$ there exists a recursive $\rho_k \colon S_k \to F_I S_k$ such that $\alpha_k = {[i_\mathsf{t}, \delta_\mathsf{t}]}^{\rho_k}$, then the final hypothesis is minimal.
\end{restatable}

\section{Generalised Tree Automata}%
\label{sec:examples}

We now instantiate the above development to a wide class of $\cat{Set}$ endofunctors. This yields an abstract algorithm for \emph{generalised} tree automata---i.e., automata accepting sets of trees, possibly subject to equations---including bottom-up tree automata and unordered tree automata.
We first introduce the running examples.

\begin{example}[Tree automata]%
\label{ex:trees}
Let $\Gamma$ be a ranked alphabet, i.e., a finite set where $\gamma \in \Gamma$ comes with $\arity(\gamma) \in \N$. The set of $\Gamma$-trees over a finite set of leaf symbols $I$ is the smallest set $T_\Gamma(I)$ such that $I \subseteq T_\Gamma(I)$, and for all $\gamma \in \Gamma$ we have that $t_1,\dots,t_{\arity(\gamma)} \in T_\Gamma(I)$ implies $(\gamma, t_1,\dots,t_{\arity(\gamma)}) \in T_\Gamma(I)$. The alphabet $\Gamma$ gives rise to the polynomial functor $FX = \coprod_{\gamma \in \Gamma} X^{\arity(\gamma)}$. The free $F$-algebra monad is precisely $T_\Gamma$, where the unit turns elements into leaves, and the multiplication flattens nested trees into a tree. A bottom-up deterministic tree automaton~\cite{hubert-etal-2008} is then an automaton over $F$, with finite $I$ and $O = 2$.
\end{example}

\begin{example}[Unordered tree automata]%
	\label{ex:unord-trees}
	Consider the finite powerset functor $\fpow \colon \cat{Set} \to \cat{Set}$, mapping a set to its finite subsets. The corresponding free $\fpow$-monad maps a set $X$ to the set of finitely-branching unordered trees with nodes in $X$. Automata over $\fpow$, with output set $O = 2$ and finite $I$, accept sets of such trees. Note that unordered trees can be seen as trees over a ranked alphabet $\Gamma = \{ \mathfrak{s}_i \mid i \in \mathbb{N}\}$, where $\arity(\mathfrak{s}_i) = i$, satisfying equations that collapse duplicate branches and identify lists of branches up to permutations.
\end{example}
Automata in these examples are algebras for endofunctors with the following properties: they are \emph{strongly finitary}~\cite{adamek2003}---i.e., they are finitary and preserve finite sets---and they preserve weak pullbacks. We turn these into a global assumption, used in several places; in particular, that $F$ is strongly finitary is used to guarantee the existence
of finite counterexamples.
\begin{assumption}
	In the remainder we take $\cat{C} = \cat{Set}$ with the (surjective, injective) factorisation system, assume $F$ is strongly finitary and preserves weak pullbacks, and $I$ finite.
\end{assumption}
If the target automaton $\aut_\mathsf{t}$ is finite, the algorithm terminates by \autoref{thm:termination}.

\tobias{In \autoref{thm:termination}, we assume that they have finitely many subobjects, is that not the same?}
\matteo{Yes. Just reminding the reader that we (need to) assume a finite target automaton.}

We start with the central notion of \emph{contextual wrapper},
a specific form of wrapper using contexts to generalise string concatenation to trees. We  then show that contextual wrappers enable effective procedures for local closedness and consistency, and for computing hypotheses. Moreover, they can always be updated via finite counterexamples.
Altogether, this makes the ingredients of our abstract algorithm concrete for generalised tree automata.

\subsection{Contextual wrappers}%
\label{ssec:contextual-wrappers}

Denote by $1$ the set $\{\square\}$.
Given $x \in X$ for any set $X$, we write $\mathfrak{e}_x$ for the function $1 \to X$ that assigns $x$ to $\square$.
We use the set $1$ to define the set of \emph{contexts} $T(I + 1)$, where the \emph{holes} $\square$ occurring in a context $c \in T(I + 1)$ can be used to plug in further data such as another context or a tree, e.g., in the case of \Cref{ex:trees,ex:unord-trees}.
In fact, it is well known that $T(I + (-))$ forms a monad with unit $\hat{\eta}_X$ turning each hole from $X$ into a context, and multiplication $\hat{\mu}_X$ plugging a context into another context~\cite{luth2002}.

\begin{definition}[Contextual wrapper]\label{def:contextual-wrapper}
 	Let $S \subseteq TI$ and $E \subseteq T(I+1)$. Now:
	\begin{itemize}[noitemsep,topsep=0pt]
		\item $\alpha_S \colon S \to Q_\mathsf{t}$ is defined as the restriction of the reachability map of $\aut_\mathsf{t}$ to $S$;
		\item $\beta_E \colon Q_\mathsf{t} \to O^E$ is defined as the function given by $\beta_E(q)(e) = (o_\mathsf{t} \circ {[i_\mathsf{t}, \mathfrak{e}_q]}^\sharp)(e)$.
	\end{itemize}
	A wrapper is called \emph{contextual} if it is of the form $(\alpha_S, \beta_E)$ for some $S$ and $E$.
\end{definition}
Intuitively, $\beta_E$ classifies states by plugging them into every context in $E$ and comparing the resulting outputs. In the DA case, contextual wrappers are equivalent to those of \autoref{ex:dawrapper}, where row labels are plugged into word contexts---i.e., words of the form $\square \cdot e$, with $e \in A^*$---to achieve string concatenation.

We now show how to compute several morphisms induced by a wrapper. These morphisms, intuitively, correspond to different parts of an observation table, and are used for (local) closedness and consistency, and to construct the hypothesis.
In particular, we show that they can be computed concretely by querying the language $\lang_{\aut_\mathsf{t}}$, i.e., via membership queries.
\begin{restatable}[Computing wrapper morphisms]{proposition}{setmorphisms}\label{prop:setmorphisms}
	Given $S \subseteq TI$ with inclusion $j \colon S \to TI$ and $E \subseteq T(I + 1)$
	with inclusion $k \colon E \rightarrow T(I+1)$, we have:
	\begin{itemize}[topsep=0pt,noitemsep]
		\item The \emph{top observation table} $\beta_E \circ \alpha_S \colon S \to O^E$,  $s \mapsto \lang_{\aut_\mathsf{t}} \circ \mu_I \circ T[\eta_I, j \circ \mathfrak{e}_s] \circ k$;
		\item The \emph{bottom observation table} $\beta_E \circ \delta_\mathsf{t} \circ F\alpha_S \colon FS \to O^E$, $t \mapsto \lang_{\aut_\mathsf{t}} \circ \mu_I \circ T[\eta_I, \gamma_I \circ Fj \circ \mathfrak{e}_t]\circ k$;
		\item The \emph{input rows} $\beta_E \circ i_\mathsf{t} \colon I \to O^E$ given by $(\beta_E \circ i_\mathsf{t})(x) = \lang_{\aut_\mathsf{t}} \circ T[\id_I, \mathfrak{e}_x]\circ k$;
		\item The \emph{row output} $o_\mathsf{t} \circ \alpha_S \colon S \to O$ given by
		$(o_\mathsf{t} \circ \alpha_S)(s) = \lang_{\aut_\mathsf{t}}(s)$.
	\end{itemize}
\end{restatable}
\begin{example}%
\label{ex:wrapper-maps}
For tree automata, a contextual wrapper is as follows: $S \subseteq T_\Gamma(I)$ is a set of $\Gamma$-trees over $I$, and $E \subseteq T_\Gamma(I + 1)$ is formed by contexts, i.e., $\Gamma$-trees where a special leaf $\square$ may occur, or equivalently $\Gamma + \square$-trees, where $\Gamma + \square$ is the signature $\Gamma$ extended with an additional constant $\square$. Plugging into a context intuitively amounts to replacing this leaf with a tree.

We now give the intuition behind the maps of \autoref{prop:setmorphisms}:
\begin{itemize}[noitemsep,topsep=0pt]
	\item The top part of the observation table has rows labelled by trees in $S$, columns by contexts in $E$, and rows are computed by plugging their tree labels into each column context and querying the language. When $E$ contains only contexts with exactly one instance of $\square$, this corresponds precisely to the observation tables of~\cite{DrewesH03,BesombesM07}.
	\item The bottom part contains rows labelled over elements of $FS = \coprod_{\gamma \in \Gamma} S^{\arity(\gamma)}$, i.e., trees obtained by adding a new root symbol to those from $S$. This generalises adding an alphabet symbol to row labels, as done in the bottom observation table of \LStar. Rows are computed as in the top part, by plugging their tree labels into contexts $E$ and querying. 
	\item The input rows are those for the leaves $I$, and the row output function queries the language for each row label.
\end{itemize}
The case of unordered trees is analogous, with a key difference: wrapper maps are now up to equations, as both $S$ and $E$ are sets of unordered trees. The corresponding observation table can be understood as containing equivalence classes of rows and columns. For instance, the bottom part has only one successor row for each set of trees in $S$, whereas in the previous case we have one successor row for each symbol $\gamma \in \Gamma$ and $\arity(\gamma)$-list of trees from $S$.
\end{example}

\paragraph*{Hypotheses.} Recall that, given a closed and consistent wrapper $(\alpha_S,\beta_E)$, the state space of the associated hypothesis is given by the image of $\beta_E \circ \alpha_S \colon S \to O^E$.
Since $S$ and $E$ are finite sets, we can compute the image of this function.
For bottom-up and unordered tree automata, as in the DA case (see \autoref{ex:dawrapper}), this image consists of distinct rows.
The initial states, outputs and dynamics of the hypothesis automaton are defined as follows:
\begin{mathpar}
i_{\hyp_\wrap}(x) = (\beta_E \circ i_\mathsf{t})(x)
\and
o_{\hyp_\wrap}(e_\wrap(s)) = (o_\mathsf{t} \circ \alpha_S)(s)
\and
\delta_{\hyp_\wrap}(F(e_\wrap)(x)) = (\beta_E \circ \delta_\mathsf{t} \circ F\alpha_S)(x).
\end{mathpar}
Closedness and consistency ensure well-definedness.
We know from \autoref{prop:setmorphisms} how to compute those functions via membership queries.

\subsection{Witnessing local closedness and consistency}\label{sec:setlocalclosedness}

We now consider local closedness and consistency. In the current setting, these amount to equality checks on finite structures, which can be performed effectively.

\begin{restatable}[Local closedness for $\cat{Set}$ automata]{lemma}{setlocalclosedness}\label{lem:setlocalclosedness}
	Given $S, S' \subseteq TI$ and $E \subseteq T(I + 1)$ such that $S \subseteq S'$, $(\alpha_{S'}, \beta_E)$ is locally closed w.r.t.\ $\alpha_S$ if there exist $k \colon I \to S'$ and $\ell \colon FS \to S'$ such that
	(1)~$\alpha_{S'} \circ k = i_\mathsf{t}$ and
    (2)~$\alpha_{S'} \circ \ell = \delta_\mathsf{t} \circ F\alpha_S$.
\end{restatable}
\begin{example}
	For bottom-up tree automata, local closedness holds if the table $(S',E)$ already contains each leaf row (equation 1), and it contains every successor row for $S$, namely $FS = \coprod_{\gamma \in \Gamma}S^{\arity(\gamma)}$ (equation 2).
    For unordered tree automata the condition is similar, and now involves successor trees in $\fpow(S)$.
\end{example}
\begin{restatable}[Local consistency for $\cat{Set}$ automata]{lemma}{setlocalconsistency}\label{lem:setlocalconsistency}
	Let $S \subseteq TI$ and $E \subseteq E' \subseteq T(I + 1)$, with $S$ finite.
    Furthermore, suppose that for $s, s' \in S$ with $(\beta_{E'} \circ \alpha_S)(s) = (\beta_{E'} \circ \alpha_S)(s')$ we have:
	(1)~$(o_\mathsf{t} \circ \alpha_S)(s) = (o_\mathsf{t} \circ \alpha_S)(s')$; and
    (2)~$\beta_E \circ \delta_\mathsf{t} \circ F(\alpha_S \circ [\id_S, \mathfrak{e}_s]) = \beta_E \circ \delta_\mathsf{t} \circ F(\alpha_S \circ [\id_S, \mathfrak{e}_{s'}])$.
    Then $\wrap = (\alpha_S, \beta_{E'})$ is locally consistent w.r.t.\ $\beta_E$.
\end{restatable}
\begin{example}%
\label{ex:tree-loc-cons}
		For bottom-up tree automata, local consistency amounts to require the following for the table for $(S,E')$. For all $s,s' \in S$ corresponding to the same row we must have: (1) $s$ and $s'$ are both accepted/rejected; (2) successor rows obtained by plugging $s$ and $s'$ into the same one-level context from $F(S + 1) = \coprod_{\gamma \in \Gamma} {(S + \{\square\})}^{\arity(\gamma)}$ are equal.

        For unordered-tree automata, we need to compare $s$ and $s'$ only when they are equationally inequivalent. Note that one-level contexts are also up to equations, which means that the position of the hole in the context is irrelevant for computing extensions of $s$ and $s'$.
\end{example}
We now develop procedures for fixing local closedness and consistency defects. First, we show that we can always extend $S$ to make the wrapper locally closed.
\begin{restatable}{proposition}{wraplocclosed}\label{prop:wraplocclosed}
	Given finite $S \subseteq TI$ and $E \subseteq T(I + 1)$, there exists a finite $S' \subseteq TI$ such that $(\alpha_{S'}, \beta_E)$ is locally closed w.r.t.\ $\alpha_S$.
	If there exists a recursive $\rho \colon S \to F_I S$ such that ${[\eta_I, \gamma_I]}^\rho \colon S \to TI$ is the inclusion, then there exists a recursive $\rho' \colon S' \to F_I S'$ such that ${[\eta_I, \gamma_I]}^{\rho'} \colon S' \to TI$ is the inclusion.
\end{restatable}
The condition of ${[\eta_I, \gamma_I]}^\rho \colon S \to TI$ being the inclusion map in the above result amounts to prefix closedness of $S$ in tree automata, see Example~\ref{ex:rec-coalg-trees} below. Further, under this condition we have that
$\alpha_S = {[i_\mathsf{t}, \delta_\mathsf{t}]}^\rho$, since the reachability map is an algebra morphism, and similarly for $\alpha_{S'}$. This is crucial to satisfy the requirements for minimality of the termination theorem.
\begin{example}%
\label{ex:rec-coalg-trees}
To better understand the above proposition, it is worth describing what recursive coalgebras are for the automata of \Cref{ex:trees,ex:unord-trees}. For bottom-up tree automata, they are coalgebras $\rho \colon S \to \coprod_{\gamma \in \Gamma} S^{\arity(\gamma)} + I$ satisfying suitable conditions. Subtree-closed subsets of $T_\Gamma(I)$ are sets of trees closed under taking subtrees. Every subtree-closed $S$ can be made into a recursive coalgebra that returns the root symbol and its arguments, if applied to a tree of non-zero depth, and a leaf otherwise. For unordered tree automata, $\rho \colon S \to \fpow S + I$ will just return the set of subtrees or a leaf.
\end{example}
The proof of \autoref{prop:wraplocclosed}, which can be found in \Cref{sec:ap2}, is constructive and describes a naive procedure
to make a table locally closed: adding all (finitely-many) successor rows to the table. For instance, in the case of tree automata, one adds rows obtained by adding a new root symbol to trees labelling rows in all possible ways, for each symbol in the alphabet.
One may optimise the algorithm by instead adding only missing rows.

We now show how to fix local consistency, by extending a finite set of column labels $E$ to a finite set $E'$ such that the resulting wrapper is locally consistent.
\begin{restatable}{proposition}{elephant}
	Given finite $S \subseteq TI$ and $E \subseteq T(I + 1)$, define $E' \subseteq T(I + 1)$ by
	$
		E' = E \; \cup \; \{(\eta_{I + 1} \circ \kappa_2)(\square)\}
			\; \cup \; \{(\hat{\mu}_1 \circ T(\id_I + c_x))(e) \mid e \in E, x \in F(S + 1)\}
	$
	where $c_x \colon 1 \to T(I + 1)$, with $c_x = \gamma_{I + 1} \circ F[T\kappa_1 \circ j, \hat{\eta}_1] \circ \mathfrak{e}_x$, and $j \colon S \to TI$ is set inclusion. It holds that $E'$ is finite and $(\alpha_S, \beta_{E'})$ is locally consistent w.r.t.\ $\beta_E$.
\end{restatable}
For tree automata, $E'$ is $E$ plus the empty context $\square$ and the trees obtained by plugging one-level contexts formed from the current row labels (see \autoref{ex:tree-loc-cons}) into all contexts in $E$. This amounts to extending columns so that \emph{all} consistency defects are fixed.
One can optimise the procedure above by incrementally adding to $E$  only those elements of $E'$ that result in new pairs of rows being distinguished.

\subsection{Finite counterexamples}\label{sec:setcounterexamples}

Finally, we show that the teacher can always supply a finite counterexample.
\begin{restatable}[Language equivalence via finite recursion]{proposition}{langequivrec}
	Given an automaton $\aut = (Q, \delta, i, o)$, we have $\lang_{\aut_\mathsf{t}} = \lang_{\aut}$ iff $\lang_{\aut_\mathsf{t}}^\rho = \lang_{\aut}^\rho$ for all recursive $\rho \colon S \to F_I S$ such that $S$ is finite.
\end{restatable}
\begin{corollary}[Finite counterexamples]\label{cor:finite}
	Given a closed and consistent wrapper $\wrap$ for $Q_\mathsf{t}$, we have $\lang_{\hyp_\wrap} \ne \lang_{\aut_\mathsf{t}}$ iff there exists a counterexample $\rho \colon S \to F_I S$ for $\wrap$ such that $S$ is finite.
\end{corollary}
\begin{example}
Recall from \autoref{ex:rec-coalg-trees} that finite recursive coalgebras for bottom-up (resp.\ unordered) tree automata are coalgebras $\rho \colon S \to \coprod_{\gamma \in \Gamma} S^{\arity(\gamma)} + I$ (resp.\ $\rho \colon S \to \fpow S + I$). Thus, finite counterexamples are recursive coalgebras of this form where $S$ is finite or, more concretely, a finite subtree-closed set of trees.
\end{example}
Given a finite counterexample, if $\alpha_S$ arises from a recursive coalgebra $\rho$ (e.g., when $S$ is prefix-closed), updating the wrapper in line~\ref{line:update} of \autoref{alg:main} can be done as follows: (1)~combine
$\rho$ with the
recursive coalgebra in \autoref{cor:finite} via a coproduct (which preserves recursiveness); (2)~take a suitable factorisation to make sure that there is an inclusion of $S$ into $TI$, and thus that the updated $\alpha_{S'}$ forms a contextual wrapper (see \autoref{lem:rec-to-contextual} for a formal justification). Concretely, the latter step amounts to removing multiple copies of rows with the same label.
Altogether, these steps take the union of the current rows with the (prefix-closed) counterexample, and guarantee that $\alpha_{S'}$ again arises from a recursive coalgebra.

\subsection{Minimality}
\autoref{thm:termination} gives sufficient conditions for minimality of
the automaton obtained from the algorithm, namely: each $\alpha$
arises from a recursive coalgebra, and the target automaton should be minimal.
For the first condition to hold, there are two parts of the algorithm that need to be implemented appropriately, as they change $\alpha$: closing the table and adding counterexamples. This can always be done: for closing the table, this follows from \autoref{prop:wraplocclosed}; for counterexamples, the strategy outlined in the previous section yields a wrapper of the desired form.
As for the second condition, a minimal automaton exists
if the functor $F$ preserves arbitrary cointersections, which is the case iff $F$ is finitary~\cite{AT89}.

\section{Related work}\label{sec:related}

This paper proceeds in the line of work on categorical automata learning started in~\cite{jasi:auto14},
and further developed in the CALF framework~\cite{HeerdtS017,heer17:lear}. CALF provides abstract definitions
of closedness, consistency, and hypothesis and several techniques to analyse and guide the development
of concrete learning algorithms. CALF operates at a high level of abstraction and previously did not include an explicit learning algorithm. We discuss two further recent categorical
approaches to learning, which make stronger assumptions than CALF in order to allow for the definition of
concrete algorithms. The present paper is a third such approach.

Barlocco et al.~\cite{BarloccoKR19} proposed an abstract algorithm for learning
coalgebras, where tests are formed by an
abstract version of coalgebraic modal logic. On the one hand, the notion of wrapper and closedness
from CALF essentially instantiate to that setting; on the other hand, the combination of logic
and coalgebra is what enables to define an actual algorithm in~\cite{BarloccoKR19}.
The current work focuses on algebras rather than coalgebras, and is orthogonal. In particular,
it covers (bottom-up) tree automata, which are outside the scope of~\cite{BarloccoKR19}.

Urbat and Schr\"{o}der proposed another categorical approach to automata learning~\cite{abs-1911-00874},
which---similarly to the work of Barlocco et al.---makes stronger assumptions than CALF
in order to define a learning algorithm. Their work focuses primarly on automata, assuming
that the systems of interest can be viewed both as algebras and coalgebras,
and the
generality comes from allowing to instantiate these in various categories. Moreover,
it allows covering algebraic recognisers in certain cases, through a reduction to automata over
a carefully constructed alphabet; this (orthogonal) extension allows covering, e.g.,
$\omega$-languages as well as tree languages. However, the reduction to automata makes this process
quite different than the approach to tree learning in the present paper: it makes use
of an automaton over all (flat) contexts, yielding an infinite alphabet, and therefore
the algorithmic aspect is not clear. The extension to an actual algorithm for learning
tree automata is mentioned as future work in~\cite{abs-1911-00874}. In the present paper,
this is achieved by learning algebras directly.

Yet another categorical approach to learning was proposed recently by Colcombet, Petrisan, and Stabile~\cite{colcombet2021}.
Here, the way automata are modelled is rather different: not as algebras or coalgebras within a category, but as functors from a structure category to an output category.
So far this has led the authors to develop an abstract automata learning algorithm that generalises algorithms for DFAs, weighted automata, and subsequential transducers.
However, as their structure category is built by generating morphisms representing words by starting with a morphism for each alphabet symbol and closing under composition, it is unclear whether this approach could cover tree automata.

Concrete algorithms for learning tree automata and languages have appeared in the literature. The inference of regular tree languages using membership and equivalence queries appeared in~\cite{DrewesH03}, extending earlier work of Sakakibara~\cite{sakakibara}.  Later,~\cite{BesombesM07} provided a learning algorithm for regular tree automata using only membership queries. The instantiated algorithm in our paper has elements (such as the use of contexts) close to the concrete algorithms. The focus of the present paper is on presenting an algebraic framework that can effectively be instantiated to recover such concrete algorithms in a modular and canonical fashion, with proofs of correctness and termination stemming from the general framework.

\section{Future Work}\label{sec:future}
This paper makes use of the free monad of a functor $F$ in the formulation of the generalised learning algorithm, and hence can only deal with quotienting in a restricted setting, namely by flat equations in the presentation of $F$. It remains an open challenge to extend the present algorithm to a setting with more general equations.
For the concrete case of pomset languages~\cite{gischer-1988,grabowski-1981} represented by \emph{bimonoids}~\cite{lodaya-weil-2000}, we note that we have succesfully instantiated the abstract algorithm described in this paper, and augmented it to include optimisations specific to the equations that hold in that setting~\cite{pomset-learning}.
In future work, we aim to extend the ideas behind these optimisations to the abstract setting, as well.

Another direction is to extend the framework with side-effects, encoded by a monad, in the style of~\cite{heer17:lear}. This would enable learning more compact automata---albeit with richer, monadic, transitions---representing languages and, as a concrete instance, provide an active learning algorithm for weighted tree automata.

\paragraph{Acknowledgements}
T.~Kapp\'{e} was partially supported by the European Union’s Horizon 2020 research and innovation programme under the Marie Sk\l{}odowska-Curie grant agreement No. 101027412 (VERLAN), as well as ERC Starting Grant 679127 (ProFoundNet).
Gerco van Heerdt, Matteo Sammartino, and Alexandra Silva were partially supported by the EPSRC Standard Grant CLeVer (EP/S028641/1).

\bibliography{main}


\appendix

\section{Example of execution of \texorpdfstring{\LStar}{L-star}}%
\label{app:lstar-ex}

We now run \autoref{alg:lstar-main} with the target language $\lang = \{w \in {\{a\}}^* \mid |w| \ne 1\}$. The algorithm starts with $S = E = \{\eword\}$; the corresponding table is depicted in Figure~\ref{fig:obs1}. It invokes Algorithm~\ref{alg:lstar-closedcons} to check closedness and consistency of this table: it is not closed, because $\ot(a)$ does not appear in the top part of the table. To fix this, it adds the word $a$ to the set $S$.
The resulting table (Figure~\ref{fig:obs2}) is closed and consistent, so the algorithm builds the hypothesis $\hyp$ shown in Figure~\ref{fig:hyp1} and poses the equivalence query $\mathsf{EQ}(\hyp)$. The teacher returns the counterexample $aaa$, which should have been accepted, and all its prefixes are added to the set $S$. The next table (Figure~\ref{fig:obs3}) is closed, but not consistent: $\ot(\eword) = \ot(aa)$, but $\ot(\eword \cdot a) \neq \ot(aa \cdot a)$. The algorithm adds $a \cdot \eword = a$ to $E$ so that $\ot(\eword)$ and $\ot(aa)$ can be distinguished, as depicted in Figure~\ref{fig:obs4}.
The table is now closed and consistent, and the new hypothesis automaton is precisely the minimal DFA accepting $\lang$ (Figure~\ref{fig:obs5}).

\captionsetup[subfigure]{justification=centering}
\begin{figure}%
\centering%
	\subfloat[]{\label{fig:obs1}%
		\begin{tabular}[b]{r | c} 
			& $\eword$ \\
			\hline 
			$\eword$ & 1 \\
			\hline 
			$a$ & 0
		\end{tabular}%
	}\qquad%
	\subfloat[]{\label{fig:obs2}%
		\begin{tabular}[b]{r | c} 
			& $\eword$ \\
			\hline 
			$\eword$ & 1 \\
			$a$ & 0 \\
			\hline 
			$aa$ & 1
		\end{tabular}%
	}\qquad%
	\subfloat[]{\label{fig:hyp1}%
		\begin{tikzpicture}[automaton]
			\node[initial,state,accepting] (q0) {};
			\node[state] (q1) [right of=q0] {};
			\node (rr) [above of=q0] {};
			\node [above of=rr] {};
			\path[->]
			(q0) edge [bend left] node {$a$} (q1)
			(q1) edge [bend left] node {$a$} (q0);
		\end{tikzpicture}
	}\qquad%
	\subfloat[]{\label{fig:obs3}%
		\begin{tabular}[b]{r | c} 
			& $\eword$ \\
			\hline 
			$\eword$ & 1 \\
			$a$ & 0 \\
			$aa$ & 1 \\
			$aaa$ & 1 \\
			\hline 
			$aaaa$ & 1 \\
		\end{tabular}%
	}\qquad%
	\subfloat[]{\label{fig:obs4}%
		\begin{tabular}[b]{r | cc} 
			& $\eword$ & $a$ \\
			\hline 
			$\eword$ & 1 & 0 \\
			$a$ & 0 & 1 \\
			$aa$ & 1 & 1 \\
			$aaa$ & 1 & 1 \\
			\hline 
			$aaaa$ & 1 & 1 \\
		\end{tabular}%
	}%
		\subfloat[]{\label{fig:obs5}%
	\begin{tikzpicture}[automaton]
		\node[initial,state,accepting] (q0) {};
		\node[state] (q1) [right of=q0] {};
		\node[state,accepting] (q2) [below of=q1] {};
		\path[->]
		(q0) edge node {$a$} (q1)
		(q1) edge node {$a$} (q2)
		(q2) edge [loop right] node {$a$} ();
	\end{tikzpicture}
	}

	\caption{Example run of \LStar on $\lang = \{w \in {\{a\}}^* \mid |w| \ne 1\}$.}%
\end{figure}

\section{Proofs for \autoref{sec:counter}}

\sublang*
\begin{proof}
	First assume that $\lang_{\aut_\mathsf{t}}^\rho = \lang_{\aut}^\rho$ for all recursive $\rho \colon S \to F_I S$.
	Note that $TI$ is the initial algebra of $F_I$; thus $[\eta, \gamma] \colon F_I TI \to TI$ has an inverse.
	One easily sees that this inverse is recursive, with the corresponding unique maps into algebras being reachability maps.
	Thus, $\lang_{\aut_\mathsf{t}} = o_\mathsf{t} \circ i_\mathsf{t}^\sharp = o \circ i^\sharp = \lang_{\aut}$.

	Conversely, assume $\lang_{\aut_\mathsf{t}} = \lang_{\aut}$.
	Given a recursive coalgebra $\rho \colon S \to F_I S$, we have that ${[i_\mathsf{t}, \delta_\mathsf{t}]}^\rho = i_\mathsf{t}^\sharp \circ {[\eta_I, \gamma_I]}^\rho$ and ${[i, \delta]}^\rho = i^\sharp \circ {[\eta_I, \gamma_I]}^\rho$ by uniqueness.
	Thus, the diagram below commutes.
    \begin{align*}
        \begin{tikzcd}[column sep=1.3cm,ampersand replacement=\&,row sep=6.5mm]
			\&
				Q_\mathsf{t} \ar[out=0,in=90]{rd}{o_\mathsf{t}} \\
			S \ar[out=90,in=180]{ru}{{[i_\mathsf{t}, \delta_\mathsf{t}]}^\rho} \ar{r}{{[\eta, \gamma]}^\rho} \ar[out=-90,in=180,swap]{rd}{{[i, \delta]}^\rho} \&
				TI \ar{u}{i_\mathsf{t}^\sharp} \ar{r}{\lang_{\aut_\mathsf{t}} = \lang_{\aut}} \ar{d}[swap]{i^\sharp} \&
				O \\
			\&
				Q \ar[out=0,in=-90,swap]{ru}{o}
		\end{tikzcd}
        \\[\dimexpr-1.5\baselineskip+\dp\strutbox]&\qedhere
	\end{align*}
\end{proof}

\rescount*
\begin{proof}
	Since the diagram below on the left commutes, we obtain a unique diagonal $h \colon H_\wrap \to H_{\wrap'}$ on the right.
	\begin{align*}
		\begin{tikzcd}[ampersand replacement=\&,column sep=1.4cm,row sep=4mm]
			S \ar{rr}{e_\wrap} \ar{d}[swap]{\kappa_1} \ar[out=-10,in=135]{rd}[pos=.6,swap]{\alpha} \&
				\&
				H_\wrap \ar{dd}{m_\wrap} \\
			S + S' \ar{r}[swap]{[\alpha, {[i_\mathsf{t}, \delta_\mathsf{t}]}^\rho]} \ar{d}[swap]{\tau_{\wrap'}^\triangleright} \&
				Q_\mathsf{t} \ar{rd}{\beta} \\
			H_{\wrap'} \ar{rr}[swap]{\tau_{\wrap'}^\triangleleft} \&
				\&
				P
		\end{tikzcd}
        &&
        \begin{tikzcd}[column sep=.5cm,ampersand replacement=\&,row sep=4mm]
            S \ar[twoheadrightarrow]{r}{e_\wrap} \ar{d}[swap]{\kappa_1} \&
                H_\wrap \ar[dashed]{ldd}[swap,pos=.4]{h} \ar{dd}{m_\wrap} \\
            S + S' \ar{d}[swap]{\tau_{\wrap'}^\triangleright} \\
            H_{\wrap'} \ar[rightarrowtail,swap]{r}{\tau_{\wrap'}^\triangleleft} \&
                P
        \end{tikzcd}
	\end{align*}
	We will show that $h$ is an automaton homomorphism.
	Noting that $\tau_{\wrap'}^\triangleleft$ is a mono and $e_\wrap$ is an epi, commutativity of the diagrams below shows that $h$ commutes with the initial states ($h \circ i_\wrap = i_{\wrap'}$) and outputs ($o_{\wrap'} \circ h = o_\wrap$).
	\begin{gather*}
		\begin{tikzcd}[ampersand replacement=\&,row sep=4mm]
			I \ar{r}{i_\wrap} \ar{rd}[swap]{i_\mathsf{t}} \ar{dd}[swap]{i_{\wrap'}} \&
				H_\wrap \ar{r}{h} \ar{rdd}{m_\wrap} \ar[phantom]{d}{\circled{1}} \ar[phantom,bend right=25]{r}{\circled{2}} \&
				H_{\wrap'} \ar[rightarrowtail]{dd}{\tau_{\wrap'}^\triangleleft} \\
			\&
				Q_\mathsf{t} \ar{rd}[swap]{\beta} \\
			H_{\wrap'} \ar[rightarrowtail,swap]{rr}{\tau_{\wrap'}^\triangleleft} \ar[phantom]{ru}{\circled{1}} \&
				\&
				P
		\end{tikzcd}
        \\
		\begin{tikzcd}[ampersand replacement=\&,row sep=5mm,column sep=14mm]
			S \ar[twoheadrightarrow]{rrr}{e_\wrap} \ar{rd}[swap]{\kappa_1} \ar[twoheadrightarrow]{dd}[swap]{e_\wrap} \ar[out=-5,in=150]{rrd}[pos=.75]{\alpha} \&
				\&
				\&
				H_\wrap \ar{dd}{o_\wrap} \\
			\&
				S + S' \ar{r}[swap]{[\alpha,{[i_\mathsf{t}, \delta_\mathsf{t}]}^\rho]} \ar{d}{\tau_{\wrap'}^\triangleright} \&
				Q_\mathsf{t} \ar{rd}{o_\mathsf{t}} \ar[phantom]{ru}{\circled{3}} \\
			H_\wrap \ar{r}[swap]{h} \ar[phantom,bend left=10]{ru}{\circled{2}} \&
				H_{\wrap'} \ar{rr}[swap]{o_{\wrap'}} \ar[phantom,bend right=10]{ru}[pos=.4]{\circled{3}} \&
				\&
				O
		\end{tikzcd}
        \\
		\begin{array}{lll}
			\circled{1}\text{ closedness} &
			\circled{2}\text{ definition of $h$} &
			\circled{3}\text{ consistency}
		\end{array}
	\end{gather*}
	As for the transition functions, we use that $\tau_{\wrap'}^\triangleleft$ is a mono to show that $h \circ \close_\wrap = \close_{\wrap'} \circ F\kappa_1$ with the commutative diagram below.
	\begin{gather*}
		\begin{tikzcd}[column sep=1.6cm,row sep=5mm,ampersand replacement=\&]
			FS \ar{r}{\close_\wrap} \ar{rd}{F\alpha} \ar{d}[swap]{F\kappa_1} \&
				H_\wrap \ar{r}{h} \ar{rddd}{m_\wrap} \ar[phantom]{d}{\circled{1}} \ar[phantom,bend right=40]{r}[pos=.6]{\circled{2}} \&
				H_{\wrap'} \ar[rightarrowtail]{ddd}{\tau_{\wrap'}^\triangleleft} \\
			F(S + S') \ar{r}[swap]{F[\alpha, {[i_\mathsf{t}, \delta_\mathsf{t}]}^\rho]} \ar{dd}[swap]{\close_{\wrap'}} \&
				FQ_\mathsf{t} \ar{d}{\delta_\mathsf{t}} \\
			\&
				Q_\mathsf{t} \ar{rd}{\beta} \\
			H_{\wrap'} \ar[rightarrowtail,swap]{rr}{\tau_{\wrap'}^\triangleleft} \ar[phantom,bend left=20]{ru}[pos=.6]{\circled{1}} \&
				\&
				P
		\end{tikzcd}
        \\
		\begin{array}{ll}
			\circled{1}\text{ closedness} &
				\circled{2}\text{ definition of $h$}
		\end{array}
	\end{gather*}
	We are now ready to show that $h \circ \delta_\wrap = \delta_{\wrap'} \circ Fh$.
	This follows from commutativity of the diagram below using the fact that  $Fe_\wrap$ is an epi.
	\begin{gather*}
		\begin{tikzcd}[column sep=1.1cm,row sep=.9cm,ampersand replacement=\&]
			FS \ar[twoheadrightarrow]{rr}{Fe_\wrap} \ar{rd}[swap]{F\kappa_1} \ar[twoheadrightarrow]{dd}[swap]{Fe_\wrap} \ar{rrd}{\close_\wrap} \&
				\&
				FH_\wrap \ar{d}{\delta_\wrap} \ar[phantom,bend right=50]{d}{\circled{1}} \\
			\&
				F(S + S') \ar{d}{F\tau_{\wrap'}^\triangleright} \ar{rd}{\close_{\wrap'}} \ar[phantom]{r}{\circled{4}} \&
				H_\wrap \ar{d}{h} \\
			FH_\wrap \ar{r}[swap]{Fh} \ar[phantom,bend left=10]{ru}[pos=.6]{\circled{2}} \&
				FH_{\wrap'} \ar{r}[swap]{\delta_{\wrap'}} \ar[phantom,bend left=25]{r}[pos=.3]{\circled{3}} \&
				H_{\wrap'}
		\end{tikzcd}
        \\
		\begin{array}{ll}
			\circled{1}\text{ definition of $\delta_\wrap$} &
            \circled{2}\text{ definition of $h$} \\
			\circled{3}\text{ definition of $\delta_{\wrap'}$} &
            \circled{4}\text{ previous observation }
		\end{array}
	\end{gather*}
	Thus, $h$ is an automaton homomorphism $\hyp_\wrap \to \hyp_{\wrap'}$.
	This implies in particular that $h \circ {[i_\wrap, \delta_\wrap]}^\rho = {[i_{\wrap'}, \delta_{\wrap'}]}^\rho$.
	It follows that the diagram below commutes.
	\begin{equation}\label{eq:uptorho}
		\begin{tikzcd}[column sep=1.2cm,row sep=5mm]
			S' \ar{r}{{[i_\wrap, \delta_\wrap]}^\rho} \ar{dr}[swap]{{[i_{\wrap'}, \delta_{\wrap'}]}^\rho} &
				H_\wrap \ar{dr}{o_\wrap} \ar{d}{h} \\
			&
				H_{\wrap'} \ar{r}[swap]{o_{\wrap'}} &
				O
		\end{tikzcd}
	\end{equation}

	We now show that $\tau_{\wrap'}^\triangleright \circ \kappa_2 = {[i_{\wrap'}, \delta_{\wrap'}]}^\rho$.
	This follows by the uniqueness property of ${[i_{\wrap'}, \delta_{\wrap'}]}^\rho$ from commutativity of the diagram below, using that $\tau_{\wrap'}^\triangleleft$ is monic.
	\begin{gather*}
		\begin{tikzcd}[column sep=1.3cm,ampersand replacement=\&,row sep=5mm]
			 \&
				F_I(S + S') \ar{r}{F_I \tau_{\wrap'}^\triangleright} \ar{d}[swap]{F_I[\alpha, {[i_\mathsf{t}, \delta_\mathsf{t}]}^\rho]} \ar{rd}[swap,sloped]{[i_{\wrap'}, \close_{\wrap'}]} \ar[phantom,bend left=15]{rd}{\circled{3}} \&
				F_I H_{\wrap'} \ar{d}{[i_{\wrap'}, \delta_{\wrap'}]} \\
			F_I S' \ar[bend left=20]{ru}{F_I\kappa_2} \ar{r}[swap]{F_I{[i_\mathsf{t}, \delta_\mathsf{t}]}^\rho} \&
				F_I Q_\mathsf{t} \ar{d}[swap]{[i_\mathsf{t}, \delta_\mathsf{t}]} \ar[phantom,bend left=15,pos=0.4]{rdd}{\circled{2}} \&
				H_{\wrap'} \ar[rightarrowtail]{dd}{\tau_{\wrap'}^\triangleleft} \\
            S' \ar{u}{\rho} \ar{r}[swap]{{[i_\mathsf{t}, \delta_\mathsf{t}]}^\rho} \ar{d}[swap]{\kappa_2} \ar[phantom]{ru}{\circled{1}} \&
				Q_\mathsf{t} \ar{rd}{\beta} \ar[phantom,bend left=40]{d}{\circled{4}} \\
			S + S' \ar{ru}[swap]{[\alpha, {[i_\mathsf{t}, \delta_\mathsf{t}]}^\rho]} \ar{r}[swap]{\tau_{\wrap'}^\triangleright} \&
				H_{\wrap'} \ar[rightarrowtail]{r}[swap]{\tau_{\wrap'}^\triangleleft} \&
				P
		\end{tikzcd}
        \\
        \begin{array}{ll}
            \circled{1}\text{ definition of ${[i_\mathsf{t}, \delta_\mathsf{t}]}^\rho$} &
            \circled{2}\text{ closedness and def.\ of $i_\wrap$ } \\
            \circled{3}\text{ definition of $\delta_{\wrap'}$ } &
            \circled{4}\text{ definitions of $\tau_{\wrap'}^\triangleright$ and $\tau_{\wrap'}^\triangleleft$ }
        \end{array}
	\end{gather*}
	The commutative diagram below completes the proof.
	\begin{gather*}
		\begin{tikzcd}[column sep=1.8cm,ampersand replacement=\&,row sep= 5mm]
			S' \ar[bend left=20]{rrd}{{[i_\mathsf{t}, \delta_\mathsf{t}]}^\rho} \ar{rd}{\kappa_2} \ar{rdd}[swap,pos=.7]{{[i_{\wrap'}, \delta_{\wrap'}]}^\rho} \ar{dd}[swap]{{[i_\wrap, \delta_\wrap]}^\rho} \\
			\&
				S + S' \ar{r}[pos=.45]{[\alpha, {[i_\mathsf{t}, \delta_\mathsf{t}]}^\rho]} \ar{d}{\tau_{\wrap'}^\triangleright} \ar[phantom,bend left=15]{rdd}[pos=.45]{\circled{1}} \ar[phantom, bend right=40]{d}[pos=0.4]{\circled{2}} \&
				Q_\mathsf{t} \ar{dd}{o_\mathsf{t}} \\
			H_\wrap \ar[bend right=20]{rrd}[swap]{o_\wrap} \ar[phantom]{r}[pos=.6]{\eqref{eq:uptorho}} \&
				H_{\wrap'} \ar{rd}[swap]{o_{\wrap'}} \\
			\&
				\&
				O
		\end{tikzcd}
        \\
        \begin{array}{ll}
            \circled{1}\text{ consistency} &
            \circled{2}\text{ previous observation }
        \end{array}
		\qedhere
	\end{gather*}
\end{proof}

\section{Proofs for Section~\ref{sec:alg}}
\label{sec:ap1}

\begin{lemma}\label{lem:selectormorphism}
	For all morphisms $\alpha_1 \colon S_1 \to Q_\mathsf{t}$, $\alpha_2 \colon S_2 \to Q_\mathsf{t}$, and $f \colon S_1 \to S_2$ such that $\alpha_2 \circ f = \alpha_1$ we have $\alpha_1^\triangleleft \le \alpha_2^\triangleleft$.
\end{lemma}
\begin{proof}
	This follows directly from the unique diagonal obtained in the commutative diagram below.
	\begin{align*}
		\begin{tikzcd}[ampersand replacement=\&]
			S_1 \ar[twoheadrightarrow]{r}{\alpha_1^\triangleright} \ar{d}[swap]{f} \&
				\bullet \ar{dd}{\alpha_1^\triangleleft} \ar[dashed]{ldd}{} \\
			S_2 \ar{d}[swap]{\alpha_2^\triangleright} \\
			\star \ar[rightarrowtail]{r}{\alpha_2^\triangleleft} \&
			Q_\mathsf{t}
		\end{tikzcd}
        \\[\dimexpr-\baselineskip+\dp\strutbox]&\qedhere
	\end{align*}
\end{proof}

\locallyclosed*
\begin{proof}
	Let $\wrap = ([\alpha, [i_\mathsf{t}, \delta_\mathsf{t}] \circ F_I \alpha], \beta)$.
	Note that $\alpha^\triangleleft \le {[\alpha, [i_\mathsf{t}, \delta_\mathsf{t}] \circ F_I \alpha]}^\triangleleft$ by \autoref{lem:selectormorphism} (via $\kappa_1 \colon S \to S + F_I S$).
	We define
	\begin{align*}
		i_\wrap &
			= I \xrightarrow{\kappa_1} F_I S \xrightarrow{\kappa_2} S + F_I S \xrightarrow{e_\wrap} H_\wrap \\
		\lclose_{\wrap,\alpha} &
			= FS \xrightarrow{\kappa_2} F_I S \xrightarrow{\kappa_2} S + F_I S \xrightarrow{e_\wrap} H_\wrap
	\end{align*}
	and note that the commutative diagrams below show that they satisfy the required properties.
	\[
		\begin{tikzcd}[ampersand replacement=\&,column sep=.9cm]
			I \ar[bend left]{rrd}{i_\mathsf{t}} \ar{rd}{\kappa_2 \circ \kappa_1} \ar{d}[swap]{\kappa_2 \circ \kappa_1} \\
			S + F_I S \ar{r}[swap]{\id + F_I\alpha} \ar{d}[swap]{e_\wrap} \&
				S + F_I Q_\mathsf{t} \ar{r}[swap]{[\alpha, [i_\mathsf{t}, \delta_\mathsf{t}]]} \&
				Q_\mathsf{t} \ar{d}{\beta} \\
			H_\wrap \ar{rr}{m_\wrap} \&
				\&
				P
		\end{tikzcd}
        \]
        \begin{align*}
			\begin{tikzcd}[ampersand replacement=\&,column sep=.9cm]
				FS \ar{r}{F_I\alpha} \ar{d}[swap]{\kappa_2 \circ \kappa_2} \&
					FQ_\mathsf{t} \ar[bend left]{rd}{\delta_\mathsf{t}} \ar{d}[swap]{\kappa_2 \circ \kappa_2} \\
				S + F_I S \ar{r}[swap]{\id + F_I\alpha} \ar{d}[swap]{e_\wrap} \&
					S + F_I Q_\mathsf{t} \ar{r}[swap]{[\alpha, [i_\mathsf{t}, \delta_\mathsf{t}]]} \&
					Q_\mathsf{t} \ar{d}{\beta} \\
				H_\wrap \ar{rr}{m_\wrap} \&
					\&
					P
			\end{tikzcd}
        \\[\dimexpr-\baselineskip+\dp\strutbox]&\qedhere
	\end{align*}
\end{proof}

\begin{lemma}\label{lem:closed}
	For any run ${\{\wrap_n = (\alpha_n, \beta_n)\}}_{n \in \N}$ and $n \in \N$, if $\alpha_{n + 1}^\triangleleft \le \alpha_n^\triangleleft$, then $\wrap_n$ is closed.
\end{lemma}
\begin{proof}
	For each $j \in \N$, denote by $S_j$ the domain of $\alpha_j$ and by $P_j$ the codomain of $\beta_j$, and let $X_j$ be the object through which $\alpha_j$ factorises.
	We define $f_j \colon X_n \to H_n$ as the unique diagonal in the commutative square below.
	\[
		\begin{tikzcd}
			S_j \ar[twoheadrightarrow]{r}{\alpha_j^\triangleright} \ar{dd}[swap]{\tau_{\wrap_j}^\triangleright} &
				X_j \ar{d}{\alpha_j^\triangleleft} \ar[dashed]{ddl}[swap]{f_j} \\
			&
				Q_\mathsf{t} \ar{d}{\beta_j} \\
			H_{\wrap_j} \ar[rightarrowtail]{r}{\tau_{\wrap_j}^\triangleleft} &
				P_j
		\end{tikzcd}
	\]
	Note that $f_j \in \epi$ because $\tau_{\wrap_j}^\triangleright \in \epi$.

	We write $v \colon X_{n + 1} \to X_n$ for the witness of $\alpha_{n + 1}^\triangleleft \le \alpha_n^\triangleleft$.
	Assume towards a contradiction that $\wrap_n$ is not closed.
	By the definition of a run we then have that $\beta_{n + 1} = \beta_n$ and $\wrap_{n + 1}$ is locally closed w.r.t.\ $\alpha_n$.
	Define $i_{\wrap_n} = h \circ i_{\wrap_{n + 1}} \colon I \to H_{\wrap_n}$ and $\close_{\wrap_n} = h \circ \lclose_{\wrap_{n + 1}, \alpha_n}$, where $i_{\wrap_{n + 1}}$ and $\lclose_{\wrap_{n + 1}, \alpha_n}$ exist by local closedness and $h \colon H_{\wrap_{n + 1}} \to H_{\wrap_n}$ is the unique diagonal in the commutative diagram below.
	\begin{align*}
		\begin{gathered}
			\begin{tikzcd}[column sep=.5cm,ampersand replacement=\&]
				X_{n + 1} \ar{rr}{f_{n + 1}} \ar{rd}{\alpha_{n + 1}^\triangleleft} \ar{d}[swap]{v} \&
					\&
					H_{\wrap_{n + 1}} \ar{dd}{\tau_{\wrap_{n + 1}}^\triangleleft} \\
				X_n \ar{r}[swap]{\alpha_n^\triangleleft} \ar{d}[swap]{f_n} \&
					Q_\mathsf{t} \ar{rd}[pos=.2]{\beta_n = \beta_{n + 1}} \ar[phantom]{ru}{\circled{1}} \\
				H_{\wrap_n} \ar{rr}[swap]{\tau_{\wrap_n}^\triangleleft} \ar[phantom,bend right=25]{ru}[pos=.6]{\circled{1}} \&
					\&
					P_n = P_{n + 1}
			\end{tikzcd}
            \\
			\begin{array}{l}
				\circled{1}\text{ definition of $f_n$ or $f_{n + 1}$}
			\end{array}
		\end{gathered} &
			&
			\begin{tikzcd}[ampersand replacement=\&]
				X_{n + 1} \ar[twoheadrightarrow]{r}[pos=.45]{f_{n + 1}} \ar{d}[swap]{v} \&
					H_{\wrap_{n + 1}} \ar{dd}{\tau_{\wrap_{n + 1}}^\triangleleft} \ar[dashed]{ddl}[swap]{h} \\
				X_n \ar{d}[swap]{f_n} \\
				H_{\wrap_n} \ar[rightarrowtail]{r}{\tau_{\wrap_n}^\triangleleft} \&
					P_n = P_{n + 1}
			\end{tikzcd}
	\end{align*}
	Now the diagrams below commute, leading to the desired contradiction that $\wrap_n$ is closed.
	\begin{align*}
		\begin{tikzcd}[ampersand replacement=\&,row sep=.6cm,column sep=0.8cm]
			I \ar{rd}{i_\mathsf{t}} \ar{d}[swap]{i_{\wrap_{n + 1}}} \\
			H_{\wrap_{n + 1}} \ar{rd}[pos=.4]{\tau_{\wrap_{n + 1}}^\triangleleft} \ar{d}[swap]{h} \ar[phantom,bend left=10]{r}[pos=.3]{\circled{1}} \ar[phantom,bend right=25]{rd}[pos=.4]{\circled{2}} \&
				Q_\mathsf{t} \ar{d}{\beta_n} \\
			H_{\wrap_n} \ar{r}[swap]{\tau_{\wrap_n}^\triangleleft} \&
				P_n = P_{n + 1}
		\end{tikzcd}
        &&
        \begin{tikzcd}[ampersand replacement=\&,row sep=.6cm,column sep=1.0cm]
            FS_n \ar{r}{F\alpha_n} \ar{d}[pos=.4]{\lclose_{\wrap_{n + 1}, \alpha_n}} \&
                FQ_\mathsf{t} \ar{d}{\delta_\mathsf{t}} \\
            H_{\wrap_{n + 1}} \ar{rd}[pos=.2]{\tau_{\wrap_{n + 1}}^\triangleleft} \ar{d}[swap]{h} \ar[phantom,bend left=20]{r}[pos=.7]{\circled{1}} \ar[phantom,bend right=25]{rd}[pos=.4]{\circled{2}} \&
                Q_\mathsf{t} \ar{d}{\beta_n} \\
            H_{\wrap_n} \ar{r}[swap]{\tau_{\wrap_n}^\triangleleft} \&
                P_n = P_{n + 1}
        \end{tikzcd}
	\end{align*}
	\[
		\begin{array}{ll}
			\circled{1}\text{ local closedness} &
			\circled{2}\text{ definition of $h$}
		\end{array}
		\qedhere
	\]
\end{proof}

\begin{proposition}\label{prop:reachable}
	Given a recursive $\rho \colon S \to F_I S$ and a closed and consistent wrapper for $Q_\mathsf{t}$ of the form $\wrap = ({[i_\mathsf{t}, \delta_\mathsf{t}]}^\rho, \beta)$, we have that $e_\wrap = {[i_\wrap, \delta_\wrap]}^\rho$ and $\hyp_\wrap$ is reachable.
\end{proposition}
\begin{proof}
	We will first show that $e_\wrap = {[i_\wrap, \delta_\wrap]}^\rho$ by using the uniqueness property of the right hand side.
	This follows from the commutative diagram below as a result of $m_\wrap$ being a mono, together with the uniqueness property of ${[i_\wrap, \delta_\wrap]}^\rho$.
	\begin{gather*}
		\begin{tikzcd}[column sep=1.2cm,ampersand replacement=\&,row sep=6mm]
			\& F_I S \ar{rr}{F_I e_\wrap} \ar{d}{F_I{[i_\mathsf{t}, \delta_\mathsf{t}]}^\rho} \ar{rrd}[swap,pos=0.8]{[i_\wrap, \close_\wrap]} \ar[phantom,bend left=20]{rrd}[pos=0.8]{\circled{3}} \ar[bend right=20,phantom]{rrdd}{\circled{2}} \&\&
				F_I H_\wrap \ar{d}{[i_\wrap, \delta_\wrap]} \\
			\&
				F_I Q_\mathsf{t} \ar{d}{[i_\mathsf{t}, \delta_\mathsf{t}]}  \&\&
				H_\wrap \ar[rightarrowtail]{d}{m_\wrap}
                \\
			S \ar{r}{{[i_\mathsf{t}, \delta_\mathsf{t}]}^\rho} \ar[bend right=10]{rrd}[swap]{e_\wrap} \ar[bend left=20]{ruu}{\rho} \ar[phantom,bend right=10]{ruu}[pos=0.4]{\circled{1}} \&
				Q_\mathsf{t} \ar{rr}{\beta} \ar[phantom]{rd}{\circled{4}} \&\&
				P \\
			\&
				\&
				H_\wrap \ar[rightarrowtail,bend right=10]{ru}[swap]{m_\wrap}
		\end{tikzcd}
        \\
		\begin{array}{ll}
			\circled{1}\text{ definition of ${[i_\mathsf{t}, \delta_\mathsf{t}]}^\rho$} &
			\circled{2}\text{ closedness and def.\ of $i_\wrap$} \\
			\circled{3}\text{ definition of $\delta_\wrap$} &
			\circled{4}\text{ definition of $\tau$} \\
		\end{array}
	\end{gather*}
	Now
	\(
		i_\wrap^\sharp \circ {[\eta_I, \gamma_I]}^\rho = {[i_\wrap, \delta_\wrap]}^\rho = e_\wrap \in \epi
	\),
	so $i_\wrap^\sharp \in \epi$ by~\cite[Proposition~14.11 via duality]{adamek2009}.
	Thus, $H_\wrap$ is reachable.
\end{proof}

\begin{lemma}\label{lem:consistent}
	For any run ${\{\wrap_n = (\alpha_n, \beta_n)\}}_{n \in \N}$ and $n \in \N$, if $\wrap_n$ is closed and $\beta_n^\triangleright \le \beta_{n + 1}^\triangleright$, then $\wrap_n$ is consistent.
\end{lemma}
\begin{proof}
	For each $j \in \N$, denote by $S_j$ the domain of $\alpha_j$ and by $P_j$ the codomain of $\beta_j$, and let $X_j$ be the object through which $\beta_j$ factorises.
	We define $f_j \colon H_{\wrap_j} \to X_j$ as the unique diagonal in the commutative square below.
	\[
		\begin{tikzcd}
			S_j \ar[twoheadrightarrow]{r}{\tau_{\wrap_j}^\triangleright} \ar{d}[swap]{\alpha_j} &
				H_{\wrap_j} \ar[dashed]{ddl}[swap]{f_j} \ar{dd}{\tau_{\wrap_j}^\triangleleft} \\
			Q_\mathsf{t} \ar{d}[swap]{\beta_j^\triangleright} \\
			X_j \ar[rightarrowtail]{r}{\beta_j^\triangleleft} &
				P_j
		\end{tikzcd}
	\]
	Note that $f_j \in \mono$ because $\tau_{\wrap_j}^\triangleleft \in \mono$.

	We write $v \colon X_n \to X_{n + 1}$ for the witness of $\beta_n^\triangleright \le \beta_{n + 1}^\triangleright$.
	Assume towards a contradiction that $\wrap_n$ is not consistent.
	By the definition of a run of the algorithm we then have that $\alpha_{n + 1} = \alpha_n$ and $\wrap_{n + 1}$ is locally consistent w.r.t.\ $\beta_n$.
	Define $o_{\wrap_n} = o_{\wrap_{n + 1}} \circ h \colon H_{\wrap_n} \to O$ and $\cons_{\wrap_n} = \lcons_{\wrap_{n + 1}, \beta_n} \circ h$, where $h \colon H_{\wrap_n} \to H_{\wrap_{n + 1}}$ is the unique diagonal in the commutative diagram below.
	\begin{align*}
		\begin{gathered}
			\begin{tikzcd}[column sep=.5cm,ampersand replacement=\&]
				S_n = S_{n + 1} \ar{rr}{\tau_{\wrap_n}^\triangleright} \ar{rd}[swap,pos=.8]{\alpha_n = \alpha_{n + 1}} \ar{dd}[swap]{\tau_{\wrap_{n + 1}}^\triangleright} \&
					\&
					H_{\wrap_n} \ar{d}{f_n} \\
				\&
					Q_\mathsf{t} \ar{r}{\beta_n^\triangleright} \ar{rd}[swap]{\beta_{n + 1}^\triangleright} \ar[phantom,bend left=25]{ru}[pos=.6]{\circled{1}} \&
					X_n \ar{d}{v} \\
				H_{\wrap_{n + 1}} \ar{rr}[pos=.3]{f_{n + 1}} \ar[phantom]{ru}{\circled{1}} \&
					\&
					X_{n + 1}
			\end{tikzcd}
            \\
			\begin{array}{l}
				\circled{1}\text{ definition of $f_n$ or $f_{n + 1}$}
			\end{array}
		\end{gathered}
        &&
        \begin{tikzcd}[column sep=.6cm,ampersand replacement=\&]
            S_n = S_{n + 1} \ar[twoheadrightarrow]{r}[pos=.3]{\tau_{\wrap_n}^\triangleright} \ar{dd}[swap]{\tau_{\wrap_{n + 1}}^\triangleright} \&
                H_{\wrap_n} \ar{d}{f_n} \ar[dashed]{ddl}[swap]{h} \\
            \&
                X_n \ar{d}{v} \\
            H_{\wrap_{n + 1}} \ar[rightarrowtail]{r}{f_{n + 1}} \&
                X_{n + 1}
        \end{tikzcd}
	\end{align*}
	Now the diagrams below commute, leading to the desired contradiction that $\wrap_n$ is consistent.
	\begin{align*}
		\begin{tikzcd}[ampersand replacement=\&,row sep=.6cm,column sep=.7cm]
			S_n = S_{n + 1} \ar{r}{\tau_{\wrap_n}^\triangleright} \ar{rd}[swap,pos=.6]{\tau_{\wrap_{n + 1}}^\triangleright} \ar{d}[swap]{\alpha_n} \ar[phantom,bend left=25]{rd}[pos=.6]{\circled{2}} \&
				H_{\wrap_n} \ar{d}{h} \\
			Q_\mathsf{t} \ar{rd}[swap]{o_\mathsf{t}} \ar[phantom,bend right=10]{r}{\circled{1}} \&
				H_{\wrap_{n + 1}} \ar{d}{o_{\wrap_{n + 1}}} \\
			\&
				O
		\end{tikzcd}
        &&
        \begin{tikzcd}[ampersand replacement=\&,row sep=.6cm,column sep=.9cm]
            FS_n = FS_{n + 1} \ar{r}{F\tau_{\wrap_n}^\triangleright} \ar{rd}[swap,pos=.8]{F\tau_{\wrap_{n + 1}}^\triangleright} \ar{d}[swap]{F\alpha_n} \ar[phantom,bend left=25]{rd}[pos=.6]{\circled{2}} \&
                FH_{\wrap_n} \ar{d}{Fv} \\
            FQ_\mathsf{t} \ar{d}[swap]{\delta_\mathsf{t}} \ar[phantom,bend right=20]{r}[pos=.3]{\circled{1}} \&
                FH_{\wrap_{n + 1}} \ar{d}[swap,pos=.6]{\lcons_{\wrap_{n + 1}, \beta_n}} \\
            Q_\mathsf{t} \ar{r}[pos=.2]{\beta_n} \&
                P_n
        \end{tikzcd}
	\end{align*}
	\[
		\begin{array}{ll}
			\circled{1}\text{ local consistency} &
			\circled{2}\text{ definition of $h$}
		\end{array}
		\qedhere
	\]
\end{proof}

\begin{lemma}\label{lem:annoying}
	Let $\alpha \colon S \to Q_\mathsf{t}$, $\alpha' \colon S' \to Q_\mathsf{t}$, and $\beta \colon Q_\mathsf{t} \to P$ be such that $\alpha^\triangleleft$ and $\alpha'^\triangleleft$ are isomorphic subobjects.
    If $(\alpha, \beta)$ is closed and consistent, then so is $(\alpha', \beta)$.
\end{lemma}
\begin{proof}
	Write $\wrap = (\alpha, \beta)$ and $\wrap' = (\alpha', \beta)$.
	Let $X$ and $X'$ be the respective objects through which $\alpha$ and $\alpha'$ factorise, and denote by $\phi \colon X \to X'$ the subobject isomorphism ($\alpha'^\triangleleft \circ \phi = \alpha^\triangleleft$).
	We define $f \colon X \to H_\wrap$ and $g \colon X' \to H_{\wrap'}$ as the unique diagonals in the diagrams below.
	\begin{align*}
		\begin{tikzcd}[ampersand replacement=\&]
			S \ar[twoheadrightarrow]{r}{\alpha^\triangleright} \ar{dd}[swap]{e_\wrap} \&
				X \ar{d}{\alpha^\triangleleft} \ar[dashed]{ddl}[swap]{f} \\
			\&
				Q_\mathsf{t} \ar{d}{\beta} \\
			H_\wrap \ar[rightarrowtail]{r}{m_\wrap} \&
				P
		\end{tikzcd}
        &&
        \begin{tikzcd}[ampersand replacement=\&]
            S' \ar[twoheadrightarrow]{r}{\alpha'^\triangleright} \ar{dd}[swap]{\tau_{\wrap'}^\triangleright} \&
                X' \ar{d}{\alpha'^\triangleleft} \ar[dashed]{ddl}[swap]{g} \\
            \&
                Q_\mathsf{t} \ar{d}{\beta} \\
            H_{\wrap'} \ar[rightarrowtail]{r}{\tau_{\wrap'}^\triangleleft} \&
                P
        \end{tikzcd}
	\end{align*}
	Note that $\alpha^\triangleright$ and $e_\wrap$ are in $\epi$, and therefore so is $f$; similarly, since $\alpha'^\triangleright$ and $\tau_{\wrap'}^\triangleright$ are in $\epi$, so is $g$~\cite[Proposition~14.9 via duality]{adamek2009}.
	We now define $\psi \colon H_\wrap \to H_{\wrap'}$ and $\psi^{-1} \colon H_\wrap \to H_{\wrap'}$ as the unique diagonals in the diagrams below.
	\begin{align*}
		\begin{tikzcd}[ampersand replacement=\&]
			X \ar[twoheadrightarrow]{r}{f} \ar{d}[swap]{\phi} \&
				H_\wrap \ar{dd}{m_\wrap} \ar[dashed]{ddl}[swap]{\psi} \\
			X' \ar{d}[swap]{g} \\
			H_{\wrap'} \ar[rightarrowtail]{r}{\tau_{\wrap'}^\triangleleft} \&
				P
		\end{tikzcd}
        &&
        \begin{tikzcd}[ampersand replacement=\&]
            X' \ar[twoheadrightarrow]{r}{g} \ar{d}[swap]{\phi^{-1}} \&
                H_{\wrap'} \ar{dd}{\tau_{\wrap'}^\triangleleft} \ar[dashed]{ddl}[swap]{\psi^{-1}} \\
            X \ar{d}[swap]{f} \\
            H_\wrap \ar[rightarrowtail]{r}{m_\wrap} \&
                P
        \end{tikzcd}
	\end{align*}
	It is a standard result that $\psi$ and $\psi^{-1}$ are inverse to each other~\cite[Proposition~14.7]{adamek2009}, as suggested by their names.
	We define
	\begin{mathpar}
		i_{\wrap'} = I \xrightarrow{i_\wrap} H_\wrap \xrightarrow{\psi} H_{\wrap'} \and o_{\wrap'} = H_{\wrap'} \xrightarrow{\psi^{-1}} H_\wrap \xrightarrow{o_\wrap} O \and d = H_{\wrap'} \xrightarrow{F\psi^{-1}} FH_\wrap \xrightarrow{\delta_\wrap} H_\wrap \xrightarrow{\psi} H_{\wrap'}.
	\end{mathpar}
	To show closedness and consistency, we will need the following two equations.
	\begin{align}\label{eq:annoying}
		o_\wrap \circ f = o_\mathsf{t} \circ \alpha^\triangleleft &
			&
			m_\wrap \circ \delta_\wrap \circ Ff = \beta \circ \delta_\mathsf{t} \circ F\alpha^\triangleleft.
	\end{align}
	Note that both $\alpha^\triangleright$ and $F\alpha^\triangleright$ are in $\epi$ because $F$ preserves $\epi$, and that they are therefore both epis.
	We use this to prove~\eqref{eq:annoying} with the commutative diagrams below.
	\begin{align*}
		\begin{gathered}
			\begin{tikzcd}[ampersand replacement=\&]
				S \ar[bend right=10]{rrd}[swap,pos=.7]{e_\wrap} \ar[twoheadrightarrow]{r}{\alpha^\triangleright} \ar[twoheadrightarrow]{d}[swap]{\alpha^\triangleright} \ar[phantom,bend left=5]{rrd}{\circled{1}} \ar[phantom]{rrdd}[pos=.7]{\circled{2}} \&
					X \ar{rd}{f} \\
				X \ar{rd}[swap]{\alpha^\triangleleft} \&
					\&
					H_\wrap \ar{d}{o_\wrap} \\
				\&
					Q_\mathsf{t} \ar{r}{o_\mathsf{t}} \&
					O
			\end{tikzcd}
            \\
			\begin{array}{l}
				\circled{1}\text{ definition of $f$} \\
				\circled{2}\text{ consistency} \\
				\circled{3}\text{ closedness} \\
				\circled{4}\text{ definition of $\delta_\wrap$}
			\end{array}
		\end{gathered} &
			&
			\begin{tikzcd}[ampersand replacement=\&]
				FS \ar[twoheadrightarrow]{r}{F\alpha^\triangleright} \ar[twoheadrightarrow]{d}[swap]{F\alpha^\triangleright} \ar[bend left=8]{rd}[swap,pos=.7]{Fe_\wrap\hspace{-.15cm}\mbox{}} \ar[bend right=10]{rdd}[swap,pos=.8]{\close_\wrap\hspace{-.1cm}\mbox{}} \ar[phantom,bend right=15]{rddd}[pos=.7]{\circled{3}} \ar[phantom,bend left=12]{rdd}[pos=.8]{\circled{4}} \ar[phantom,bend left=30]{rd}[pos=.7]{\circled{1}} \&
					FX \ar{d}{Ff} \\
				FX \ar{d}[swap]{F\alpha^\triangleleft} \&
					FH_\wrap \ar{d}{\delta_\wrap} \\
				FQ_\mathsf{t} \ar{d}[swap]{\delta_\mathsf{t}} \&
					H_\wrap \ar{d}{m_\wrap} \\
				Q_\mathsf{t} \ar{r}{\beta} \&
					P
			\end{tikzcd}
	\end{align*}
	Now the diagrams below commute.
	\begin{align*}
		\begin{tikzcd}[ampersand replacement=\&,column sep=.6cm]
			I \ar{r}{i_\mathsf{t}} \ar{d}{i_\wrap} \ar[bend right=40]{dd}[swap,pos=.2]{i_{\wrap'}} \ar[phantom,bend right=18]{dd}[pos=.75]{\circled{1}} \&
				Q_\mathsf{t} \ar{dd}{\beta} \\
			H_\wrap \ar{rd}{m_\wrap} \ar{d}{\psi} \ar[phantom,bend right=10]{ur}{\circled{2}} \ar[phantom,bend left=50]{d}[pos=.8]{\circled{3}} \\
			H_{\wrap'} \ar{r}[swap]{\tau_{\wrap'^\triangleleft}} \&
				P
		\end{tikzcd}
        &&
        \begin{tikzcd}[ampersand replacement=\&,column sep=.8cm]
            S' \ar[bend left=20]{rd}{\tau_{\wrap'}^\triangleright} \ar{d}[swap]{\alpha'^\triangleright} \ar[phantom]{rd}{\circled{4}} \\
            X' \ar{r}{g} \ar{d}{\phi^{-1}} \ar[bend right=40]{dd}[swap]{\alpha'^\triangleleft} \ar[phantom,bend left=10]{rd}[pos=.4]{\circled{6}} \ar[phantom,bend right=18]{dd}[pos=.75]{\circled{5}} \&
                H_{\wrap'} \ar{d}[swap]{\psi^{-1}} \ar[bend left=40]{dd}[pos=.2]{o_{\wrap'}} \ar[phantom,bend left=18]{dd}[pos=.75]{\circled{7}} \\
            X \ar{r}{f} \ar{d}{\alpha^\triangleleft} \ar[phantom,bend left=10]{rd}[pos=.4]{\eqref{eq:annoying}} \&
                H_\wrap \ar{d}[swap]{o_\wrap} \\
            Q_\mathsf{t} \ar{r}{o_\mathsf{t}} \&
                O
        \end{tikzcd}
	\end{align*}
	\begin{gather*}
		\begin{array}{ll}
			\circled{1}\text{ definition of $i_{\wrap'}$} &
				\circled{5}\text{ subobject morphism } \\
			\circled{2}\text{ closedness} &
				\circled{6}\text{ definition of $\psi^{-1}$} \\
			\circled{3}\text{ definition of $\psi$} &
				\circled{7}\text{ definition of $o_{\wrap'}$} \\
			\circled{4}\text{ definition of $g$} &
				\circled{8}\text{ definition of $d$}
		\end{array} \\
		\begin{tikzcd}[ampersand replacement=\&,column sep=.5cm]
			FS' \ar{r}{F\alpha'^\triangleright} \ar{rd}[swap]{F\tau_{\wrap'}^\triangleright} \ar[phantom,bend right=40]{r}[pos=.9]{\circled{4}} \&
				FX' \ar{rd}{F\phi^{-1}} \ar[bend left=20]{rrd}{F\alpha'^\triangleleft} \ar{d}{Fg} \ar[phantom]{rrd}{\circled{5}} \ar[phantom,bend right=35]{rd}[pos=.65]{\circled{6}} \\
			\&
				FH_{\wrap'} \ar{rd}{\psi^{-1}} \ar[bend right=30]{rrdd}[swap]{d} \&
				FX \ar{r}{F\alpha^\triangleleft} \ar{d}{Ff} \&
				FQ_\mathsf{t} \ar{r}{\delta_\mathsf{t}} \ar[phantom]{d}{\eqref{eq:annoying}} \&
				Q_\mathsf{t} \ar{dd}{\beta} \\
			\&
				\&
				FH_\wrap \ar{r}{\delta_\wrap} \ar[phantom]{rd}{\circled{8}} \&
				H_\wrap \ar{rd}{m_\wrap} \ar{d}[swap]{\psi} \ar[phantom,bend left=25]{d}{\circled{3}} \\
			\&
				\&
				\&
				H_{\wrap'} \ar{r}[pos=.4]{\tau_{\wrap'}^\triangleleft} \&
				P
		\end{tikzcd}
	\end{gather*}
	Using~\cite[Theorem~9]{HeerdtS017}, the existence of a $d$ making the last diagram above commute shows together with the other two commutative diagrams that $\wrap'$ is closed and consistent.
\end{proof}

\alphabeta*
\begin{proof}
	We consider each of the cases listed in the definition of a run of the algorithm.
	If $\wrap_n$ is not closed, then $\beta_{n + 1} = \beta_n$ and $\alpha_n^\triangleleft \le \alpha_{n + 1}^\triangleleft$ by the definition of local closedness.
	Supposing $\alpha_{n + 1}^\triangleleft \le \alpha_n^\triangleleft$ leads by \autoref{lem:closed} to the contradiction that $\wrap_n$ is closed.

	If $\wrap_n$ is closed but not consistent, then $\alpha_{n + 1} = \alpha_n$ and we have $\beta_{n + 1}^\triangleright \le \beta_n^\triangleright$ by the definition of local consistency.
	Supposing $\beta_n^\triangleright \le \beta_{n + 1}^\triangleright$ leads by \autoref{lem:consistent} to the contradiction that $\wrap_n$ is consistent.

	If $\wrap_n$ is closed and consistent and we obtain a counterexample $\rho \colon S \to F_I S$ for $\wrap_n$, then $\beta_{n + 1} = \beta_n$ and $\alpha_{n + 1}^\triangleleft = {[\alpha_n, {[i_\mathsf{t}, \delta_\mathsf{t}]}^\rho]}^\triangleleft$.
	We have $\alpha_n^\triangleleft \le {[\alpha_n, {[i_\mathsf{t}, \delta_\mathsf{t}]}^\rho]}^\triangleleft$ using \autoref{lem:selectormorphism}.
	Suppose $\alpha_{n + 1}^\triangleleft \le \alpha_n^\triangleleft$.
	Then ${[\alpha_n, {[i_\mathsf{t}, \delta_\mathsf{t}]}^\rho]}^\triangleleft = \alpha_{n + 1}^\triangleleft \le \alpha_n^\triangleleft$, so $\alpha_n^\triangleleft$ and ${[\alpha_n, {[i_\mathsf{t}, \delta_\mathsf{t}]}^\rho]}^\triangleleft$ are isomorphic subobjects.
	By \autoref{lem:annoying} this implies that $([\alpha_n, {[i_\mathsf{t}, \delta_\mathsf{t}]}^\rho], \beta_n)$ is also closed and consistent, which by \autoref{thm:rescount} contradicts the fact that $\rho$ is a counterexample for $\wrap_n$.

	If $\wrap_n$ is closed and consistent and correct up to all recursive $F_I$-coalgebras, then we immediately have $\alpha_{n + 1} = \alpha_n$ and $\beta_{n + 1} = \beta_n$.
\end{proof}

\termination*

\begin{proof}
	We will show that ${\{\wrap_n\}}_{n \in \N}$ converges, for which it suffices to show that both ${\{\alpha_n\}}_{n \in \N}$ and ${\{\beta_n\}}_{n \in \N}$ converge.
	Suppose ${\{\alpha_n\}}_{n \in \N}$ does not converge.
	By \autoref{lem:alphabeta} there exist $i_n \in \N$ for all $n \in \N$ such that $i_{n + 1} \gt i_n$, $\alpha_{i_n}^\triangleleft \le \alpha_{i_n + 1}^\triangleleft$, and $\alpha_{i_n + 1}^\triangleleft \not\le \alpha_{i_n}^\triangleleft$ for all $n \in \N$.
	Note that isomorphic subobjects are ordered in both directions.
	Using transitivity of the order on subobjects we know that for all $m, n \in \N$ with $m \ne n$ we have that $\alpha_{i_m}^\triangleleft$ and $\alpha_{i_n}^\triangleleft$ are not isomorphic subobjects.
	This contradicts the fact that $Q_\mathsf{t}$ has finitely many subobject isomorphism classes.
	Thus, ${\{\alpha_n\}}_{n \in \N}$ must converge.

	Now suppose ${\{\beta_n\}}_{n \in \N}$ does not converge.
	By \autoref{lem:alphabeta} there exist $i_n \in \N$ for all $n \in \N$ such that $i_{n + 1} \gt i_n$, $\beta_{i_n + 1}^\triangleright \le \beta_{i_n}^\triangleright$, and $\beta_{i_n}^\triangleright \not\le \beta_{i_n + 1}^\triangleright$ for all $n \in \N$.
	Note that isomorphic quotients are ordered in both directions.
	Using transitivity of the order on quotients we know that for all $m, n \in \N$ with $m \ne n$ the quotients $\beta_{i_m}^\triangleleft$ and $\beta_{i_n}^\triangleleft$ are not isomorphic.
	This contradicts the fact that $Q_\mathsf{t}$ has finitely many quotient isomorphism classes.
	Thus, ${\{\beta_n\}}_{n \in \N}$ must converge.
	We conclude that ${\{\wrap_n\}}_{n \in \N}$ converges, and by  \autoref{prop:converge} the algorithm terminates.
	By definition, it does so with a correct hypothesis.

	Now assume that $\aut_\mathsf{t}$ is minimal and that for all $k \in \N$ there exists $\rho_k \colon S_k \to F_I S_k$ such that $\alpha_k = {[i_\mathsf{t}, \delta_\mathsf{t}]}^{\rho_k}$.
	Let $n \in \N$ be such that $\wrap_{n + 1} = \wrap_n$, which by the above we know exists, and define $\wrap = \wrap_n$.
	We know from \autoref{prop:reachable} that $\hyp_\wrap$ is reachable.
	Together with correctness of $\hyp_\wrap$ and minimality of $\aut_\mathsf{t}$ there exists a unique automaton homomorphism $h \colon \hyp_\wrap \to \aut_\mathsf{t}$.
	We show that $\beta_n \circ h = m_\wrap$ with the commutative diagram below, where we precompose with the epi $e_\wrap$ and use that automaton homomorphisms commute with restricted reachability maps.
	\[
		\begin{tikzcd}[row sep=6mm]
			S_n \ar[twoheadrightarrow]{rr}{e_\wrap} \ar[twoheadrightarrow]{d}[swap]{e_\wrap = [i_\wrap, \delta_\wrap]^{\rho_n}} \ar{rd}{\alpha_n = [i_\mathsf{t}, \delta_\mathsf{t}]^{\rho_n}} &
				&
				H_\wrap \ar{d}{m_\wrap} \\
			H_\wrap \ar{r}{h} &
				Q_\mathsf{t} \ar{r}{\beta_n} &
				P_n
		\end{tikzcd}
	\]
	It follows that $h \in \mono$~\cite[Proposition~14.11]{adamek2009}.
	Being an automaton homomorphism, $h$ commutes with the reachability maps: $h \circ i_\wrap^\sharp = i_\mathsf{t}^\sharp$.
	Because $i_\wrap^\sharp \in \epi$ and $i_\mathsf{t}^\sharp \in \epi$, we have $h \in \epi$~\cite[Proposition~14.9 via duality]{adamek2009}.
	Since $\epi \cap \mono$ contains only isomorphisms~\cite[Proposition~14.6]{adamek2009}, it follows that $h$ is an isomorphism of automata and therefore that the hypothesis is minimal.
\end{proof}

\section{Proofs for Section~\ref{sec:examples}}
\label{sec:ap2}

In some of the proofs below we will use the basic property that
for all functions $f \colon X \to Y$ and elements $x \in X$ we have:
\begin{equation}\label{eq:one}
	\mathfrak{e}_{f(x)} = f \circ \mathfrak{e}_x.
\end{equation}

\begin{lemma}\label{lem:doublecontext}
	For all $T$-algebras $(X, x)$, $p \colon I \to X$, and $c \in T(I + X)$, the diagram below commutes.
	\[
		\begin{tikzcd}[ampersand replacement=\&]
			T(I + T(I + X)) \ar{r}{\hat{\mu}_X} \&
				T(I + X) \ar{d}{[p, \id_X]^\sharp} \\
			T(I + 1) \ar{u}{T(\id_I + \mathfrak{e}_c)} \ar{r}{[p, \mathfrak{e}_{[p, \id_X]^\sharp(c)}]^\sharp} \&
				X
		\end{tikzcd}
	\]
\end{lemma}
\begin{proof}
	Given any set $Y$ and a $T(I + (-))$-algebra $(Z, z)$, the extension of a morphism $f \colon Y \to Z$ to the $T(I + (-))$-algebra homomorphsm $f^\natural \colon T(I + Y) \to Z$ is given by $f^\natural = z \circ T(\id_I + f)$.
	We can supply $X$ with the $T(I + (-))$-algebra structure ${[p, \id_X]}^\sharp \colon T(I + X) \to X$.
	Thus,
	\begin{align*}
		{[p, \mathfrak{e}_{{[p, \id_X]}^\sharp(c)}]}^\sharp &
			= {[p, \mathfrak{e}_{\id_X^\natural(c)}]}^\sharp \\
		&
			= {[p, \id_X]}^\sharp \circ T(\id_I + \mathfrak{e}_{\id_X^\natural(c)}) \\
		&
			= \mathfrak{e}_{\id_X^\natural(c)}^\natural \\
		&
			= {(\id_X^\natural \circ \mathfrak{e}_c)}^\natural &
			&
			\text{\eqref{eq:one}} \\
		&
			= \id_X^\natural \circ \mathfrak{e}_c^\natural \\
		&
			= {[p, \id_X]}^\sharp \circ \mathfrak{e}_c^\natural \\
		&
			= {[p, \id_X]}^\sharp \circ \hat{\mu}_X \circ T(\id_I + \mathfrak{e}_c). &
		&
		\qedhere
	\end{align*}
\end{proof}

\setmorphisms*
\begin{proof}
    We first claim that
    \begin{equation}\label{eq:handy}
		(\beta_E \circ i_\mathsf{t}^\sharp)(s) = \lang_{\aut_\mathsf{t}} \circ \mu_I \circ T[\eta_I, \mathfrak{e}_s].
    \end{equation}
    (Here, and below, we omit the inclusion $k$.)
    To see this, first note that
    \begin{align*}
    (\beta_E \circ i_\mathsf{t}^\sharp)(s)
        &= o_\mathsf{t} \circ {[i_\mathsf{t}, \mathfrak{e}_{i_\mathsf{t}^\sharp(s)}]}^\sharp \\
        &= o_\mathsf{t} \circ \delta_\mathsf{t}^* \circ T[i_\mathsf{t}, \mathfrak{e}_{i_\mathsf{t}^\sharp(s)}] \\
        &= o_\mathsf{t} \circ \delta_\mathsf{t}^* \circ T[i_\mathsf{t}, i_\mathsf{t}^\sharp \circ \mathfrak{e}_s] \\
        &= o_\mathsf{t} \circ \delta_\mathsf{t}^* \circ T[i_\mathsf{t}^\sharp \circ \eta_I, i_\mathsf{t}^\sharp \circ \mathfrak{e}_s] \\
        &= o_\mathsf{t} \circ \delta_\mathsf{t}^* \circ Ti_\mathsf{t}^\sharp \circ T[\eta_I, \mathfrak{e}_s]
    \end{align*}
	It remains to show that $o_\mathsf{t} \circ \delta_\mathsf{t}^* \circ Ti_\mathsf{t}^\sharp = \lang_{\aut_\mathsf{t}} \circ \mu_I$, which follows by commutativity of the diagram below.
	\begin{gather*}
		\begin{tikzcd}[ampersand replacement=\&,column sep=.6cm]
			T^2 I \ar{rd}{T^2 i_\mathsf{t}} \ar{rrrr}{\mu_I} \ar{dd}[swap]{Ti_\mathsf{t}^\sharp} \&
				\&
				\ar[phantom]{d}[pos=0.3]{\circled{3}} \&
				\&
				TI \ar{dd}{i_\mathsf{t}^\sharp} \ar{ld}[swap]{Ti_\mathsf{t}} \ar[bend left]{rdd}{\lang_{\aut_\mathsf{t}}} \\
			\ar[phantom]{r}{\circled{1}} \&
				T^2 Q_\mathsf{t} \ar{rr}{\mu_{Q_\mathsf{t}}} \ar{ld}{T\delta_\mathsf{t}^*} \&
				\ar[phantom]{d}[pos=0.3]{\circled{4}} \&
				TQ_\mathsf{t} \ar{rd}[swap]{\delta_\mathsf{t}^*} \ar[phantom]{r}{\circled{1}} \&
				\ar[phantom]{r}{\circled{2}} \&
				~ \\
            TQ_\mathsf{t} \ar{rrrr}{\delta_\mathsf{t}^*} \&
				\&
				~ \&
				\&
				Q_\mathsf{t} \ar{r}[swap]{o_\mathsf{t}} \&
				O
        \end{tikzcd}
        \\
        \begin{array}{ll}
            \circled{1}\text{ property of $i_\mathsf{t}^\sharp$ } &
			\circled{2}\text{ definition of $\lang_{\aut_\mathsf{t}}$ } \\
            \circled{3}\text{ naturality } &
            \circled{4}\text{ $(T, \delta_\mathsf{t}^*)$ is a $T$-algebra }
        \end{array}
	\end{gather*}
    For the first equation, we derive
    \begin{align*}
    (\beta_E \circ \alpha_S)(s)
        &= (\beta_E \circ i_\mathsf{t}^\sharp)(s)
            \tag{def. of $\alpha_S$} \\
		&= \lang_{\aut_\mathsf{t}} \circ \mu_I \circ T[\eta_I, \mathfrak{e}_s]
            \tag{by~\eqref{eq:handy}}
    \end{align*}
    For the second equation, we derive
	\begin{align*}
	(\beta_E \circ \delta_\mathsf{t} \circ F\alpha_S)(f)
	    &= (\beta_E \circ \delta_\mathsf{t} \circ Fi_\mathsf{t}^\sharp \circ Fj)(f)
            \tag{def. of $\alpha_S$} \\
	    &= (\beta_E \circ i_\mathsf{t}^\sharp \circ \gamma_I \circ Fj)(f)
            \tag{$i^\sharp$ is an $F$-algebra homomorphism} \\
		&= \lang_{\aut_\mathsf{t}} \circ \mu_I \circ T[\eta_I, \mathfrak{e}_{(\gamma_I \circ Fj)(f)}]
            \tag{by~\eqref{eq:handy}} \\
		&= \lang_{\aut_\mathsf{t}} \circ \mu_I \circ T[\eta_I, \gamma_I \circ Fj \circ \mathfrak{e}_f]
            \tag{by~\eqref{eq:one}} \\
	\end{align*}
    For the third equation, we derive
    \begin{align*}
    (\beta_E \circ i_\mathsf{t})(s)
        &= (\beta_E \circ i_\mathsf{t}^\sharp \circ \eta_I)(s)
            \tag{property of $i^\sharp$} \\
		&= \lang_{\aut_\mathsf{t}} \circ \mu_I \circ T[\eta_I, \mathfrak{e}_{\eta_I(s)}]
            \tag{by~\eqref{eq:handy}} \\
        &= \lang_{\aut_\mathsf{t}} \circ \mu_I \circ T[\eta_I, \eta_I \circ \mathfrak{e}_s]
            \tag{by~\eqref{eq:one}} \\
        &= \lang_{\aut_\mathsf{t}} \circ \mu_I \circ T\eta_I \circ T[\id_I, \mathfrak{e}_s] \\
        &= \lang_{\aut_\mathsf{t}} \circ T[\id_I, \mathfrak{e}_s]
            \tag{monad law}
    \end{align*}
    Finally, for the fourth equation, we derive:
	\begin{align*}
    (o_\mathsf{t} \circ \alpha_S)(s)
        &= (o_\mathsf{t} \circ i_\mathsf{t}^\sharp)(s)
           \tag{definition of $\alpha_S$} \\
		&= \lang_{\aut_\mathsf{t}}(s) &
			\tag*{(definition of $\lang_{\aut_\mathsf{t}}$) \qedhere}
	\end{align*}
\end{proof}

\setlocalclosedness*
\begin{proof}
	Let $\wrap = (\alpha_{S'}, \beta_E)$, and choose
    \begin{mathpar}
    i_\wrap = e_\wrap \circ k
    \and
    \lclose_{\wrap, \alpha_S} = e_\wrap \circ \ell
    \end{mathpar}
    The necessary diagrams now commute:
	\begin{align*}
		\begin{tikzcd}[column sep=0.7cm,ampersand replacement=\&]
			I \ar{rr}{i_\mathsf{t}} \ar{dd}[swap]{i_\wrap} \ar{rd}{k} \& \&
				Q_\mathsf{t} \ar{dd}{\beta} \\
            \& S' \ar{ld}[swap]{e_\wrap} \ar{ru}{\alpha_{S'}} \\
			H_\wrap \ar{rr}{m_\wrap} \& \&
				P
		\end{tikzcd}
        &&
        \begin{tikzcd}[column sep=.7cm,ampersand replacement=\&]
            FS' \ar{r}{F\alpha'} \ar{dd}[swap]{\lclose_{\wrap,\alpha'}} \ar{rd}{\ell} \&
                FQ_\mathsf{t} \ar{r}{\delta_\mathsf{t}} \&
                Q_\mathsf{t} \ar{dd}{\beta} \\
        \& S' \ar{ld}[swap]{e_\wrap} \ar{ru}{\alpha_{S'}} \\
            H_\wrap \ar{rr}{m_\wrap} \&
                \&
                P
        \end{tikzcd}
	\end{align*}
	Note that $\alpha_S \le \alpha_{S'}$ because $S \subseteq S'$.
	Thus, $\wrap$ is locally closed w.r.t.\ $\alpha_S$.
\end{proof}

\setlocalconsistency*
\begin{proof}
	Since $e_\wrap$ is surjective, so is $Fe_\wrap$.
	We define the function
    \[
        \lcons_{\wrap,\beta_E} \colon FH_\wrap \to O^E, \qquad
		\lcons_{\wrap,\beta_E}(Fe_\wrap(y)) = (\beta_E \circ \delta_\mathsf{t} \circ F\alpha_S)(y).
	\]
	By definition this satisfies the local consistency condition.
	It remains to show that the function is well-defined.
	Denote by $K$ the kernel
	\[
		\{(s, s') \mid s, s' \in S, e_\wrap(s) = e_\wrap(s')\}
	\]
	and let $j \colon K \to S \times S$ be the inclusion.
	Consider $y, z \in FS$ such that $Fe_\wrap(y) = Fe_\wrap(z)$.
	Because $F$ preserves weak pullbacks we can find $x \in FK$ such that $F(\pi_1 \circ j)(x) = y$ and $F(\pi_2 \circ j)(x) = z$.
	Using that $S$ is finite, write $K = \{(s_1, s_1'), \ldots, (s_n, s_n')\}$.
	For all $1 \le m \le n$ we define $f_m \colon K \to S + 1$ by
	\[
		f_m(s_k, s_k') = \begin{cases}
			\kappa_1(s_k') &
				\text{if $k \lt m$} \\
			\kappa_2(\square) &
				\text{if $k = m$} \\
			\kappa_1(s_k) &
				\text{if $k \gt m$.}
		\end{cases}
	\]
	Furthermore, let $c_m = F(f_m)(x) \in F(S + 1)$.
	We will prove that
	\begin{equation}\label{eq:jump}
		(\beta_E \circ \delta_\mathsf{t} \circ F(\alpha_S \circ [\id_S, \mathfrak{e}_{s_1}]))(c_1) = (\beta_E \circ \delta_\mathsf{t} \circ F(\alpha_S \circ [\id_S, \mathfrak{e}_{s_n}]))(c_n),
	\end{equation}
	for which it suffices by induction to prove for all $1 \le m \lt n$ that
	\begin{gather*}
		(\beta_E \circ \delta_\mathsf{t} \circ F(\alpha_S \circ [\id_S, \mathfrak{e}_{s_m}]))(c_m) \\
		= (\beta_E \circ \delta_\mathsf{t} \circ F(\alpha_S \circ [\id_S, \mathfrak{e}_{s_{m + 1}}]))(c_{m + 1}).
	\end{gather*}
	Note that
	\[
		[\id_S, \mathfrak{e}_{s_m'}] \circ f_m = [\id_S, \mathfrak{e}_{s_{m + 1}}] \circ f_{m + 1}
	\]
	by the definitions of $f_m$ and $f_{m + 1}$, so
	\begin{align*}
		&
			\phantom{{} = {}} (\beta_E \circ \delta_\mathsf{t} \circ F(\alpha_S \circ [\id_S, \mathfrak{e}_{s_m}]))(c_m) \\
		&
			= (\beta_E \circ \delta_\mathsf{t} \circ F(\alpha_S \circ [\id_S, \mathfrak{e}_{s_m'}]))(c_m) \\
		&
			\phantom{{} = {}} \text{(assumption)} \\
		&
			= (\beta_E \circ \delta_\mathsf{t} \circ F(\alpha_S \circ [\id_S, \mathfrak{e}_{s_m'}] \circ f_m))(x) \\
		&
			\phantom{{} = {}} \text{(definition of $c_m$)} \\
		&
			= (\beta_E \circ \delta_\mathsf{t} \circ F(\alpha_S \circ [\id_S, \mathfrak{e}_{s_{m + 1}}] \circ f_{m + 1}))(x) \\
		&
			= (\beta_E \circ \delta_\mathsf{t} \circ F(\alpha_S \circ [\id_S, \mathfrak{e}_{s_{m + 1}}]))(c_{m + 1}) \\
		&
			\phantom{{} = {}} \text{(definition of $c_{m + 1}$)}.
	\end{align*}
	Then
	\begin{align*}
		(\beta_E \circ \delta_\mathsf{t} \circ F\alpha_S)(y) &
			= (\beta_E \circ \delta_\mathsf{t} \circ F(\alpha_S \circ \pi_1 \circ j))(x) \\
		&
			\phantom{{} = {}} \text{(definition of $x$)} \\
		&
			= (\beta_E \circ \delta_\mathsf{t} \circ F(\alpha_S \circ [\id_S, \mathfrak{e}_{s_1}] \circ f_1))(x) \\
		&
			\phantom{{} = {}} \text{(definition of $f_1$)} \\
		&
			= (\beta_E \circ \delta_\mathsf{t} \circ F(\alpha_S \circ [\id_S, \mathfrak{e}_{s_1}]))(c_1) \\
		&
			\phantom{{} = {}} \text{(definition of $c_1$)} \\
		&
			= (\beta_E \circ \delta_\mathsf{t} \circ F(\alpha_S \circ [\id_S, \mathfrak{e}_{s_n}]))(c_n) \\
		&
			\phantom{{} = {}} \text{\eqref{eq:jump}} \\
		&
			= (\beta_E \circ \delta_\mathsf{t} \circ F(\alpha_S \circ [\id_S, \mathfrak{e}_{s_n'}]))(c_n) \\
		&
			\phantom{{} = {}} \text{(assumption)} \\
		&
			= (\beta_E \circ \delta_\mathsf{t} \circ F(\alpha_S \circ [\id_S, \mathfrak{e}_{s_n'}]) \circ f_n)(x) \\
		&
			\phantom{{} = {}} \text{(definition of $c_n$)} \\
		&
			= (\beta_E \circ \delta_\mathsf{t} \circ F(\alpha_S \circ \pi_2 \circ j))(x) \\
		&
			\phantom{{} = {}} \text{(definition of $f_n$)} \\
		&
			= (\beta_E \circ \delta_\mathsf{t} \circ F\alpha_S)(z) \\
		&
			\phantom{{} = {}} \text{(definition of $x$)}.
	\end{align*}
	We conclude that $\lcons_{\wrap,\beta_E}$ is well-defined.

	We define $o_\wrap \colon H_\wrap \to O$ by
	\[
		o_\wrap(e_\wrap(s)) = (o_\mathsf{t} \circ \alpha_S)(s).
	\]
	Again the local consistency condition is satisfied by definition, but we need to show that the function is well-defined.
	Consider $s_1, s_2 \in S$ such that $e_\wrap(s_1) = e_\wrap(s_2)$.
	Then
	\begin{align*}
		(\beta_{E'} \circ \alpha_S)(s_1) &
			= (m_\wrap \circ e_\wrap)(s_1) \\
		&
			= (m_\wrap \circ e_\wrap)(s_2) \\
		&
			= (\beta_{E'} \circ \alpha_S)(s_2),
	\end{align*}
	so $(o_\mathsf{t} \circ \alpha_S)(s_1) = (o_\mathsf{t} \circ \alpha_S)(s_2)$.
	Note that $\beta_{E'} \le \beta_E$ because $E \subseteq E'$.
	Thus, $\wrap$ is locally consistent w.r.t.\ $\beta_E$.
\end{proof}

\wraplocclosed*
\begin{proof}
	Let $j \colon S \to TI$ be the inclusion map and define
	\[
		S' = S \cup \{([\eta_I, \gamma_I] \circ F_I j)(x) \mid x \in F_I S\} \subseteq TI.
	\]
	Since $S$ and $I$ are finite and $F$ preserves finite sets, $S'$ is also finite.
	We choose $k \colon I \to S'$ and $\ell \colon FS \to S'$ by setting
    \begin{mathpar}
    k(x) = \eta_I(x)
    \and
    \ell(x) = (\gamma_I \circ Fj)(x)
    \end{mathpar}
    Note that $k$ and $\ell$ are well-defined by construction of $S'$.
    Using the definitions of $\alpha_{S'}$ and $k$, we can then derive that
    \[
    (\alpha_{S'} \circ k)(x)
        = i_\mathsf{t}^\sharp(k(x))
        = i_\mathsf{t}^\sharp(\eta_I(x))
        = i_\mathsf{t}(x)
    \]
    Furthermore, we find that
    \begin{align*}
    (\alpha_{S'} \circ \ell)(x)
		&= i_\mathsf{t}^\sharp(\ell(x))
			\tag{def. of $\alpha_{S'}$} \\
        &= i_\mathsf{t}^\sharp(\gamma_I(Fj(x)))
            \tag{def. of $\ell$} \\
        &= \delta_\mathsf{t}(F(i_\mathsf{t}^\sharp)(Fj(x)))
            \tag{$i_\mathsf{t}^\sharp$ is an $F$-algebra homomorphism} \\
        &= \delta_\mathsf{t}(F(i_\mathsf{t}^\sharp \circ j)(x))  \\
        &= (\delta_\mathsf{t} \circ F\alpha_S)(x)
            \tag{def. of $\alpha_S$}
    \end{align*}
	Hence $(\alpha_{S'}, \beta_E)$ is locally closed w.r.t.\ $\alpha_S$, by \autoref{lem:setlocalclosedness}.

	Given a recursive $\rho \colon S \to F_I S$ such that ${[\eta_I, \gamma_I]}^\rho \colon S \to TI$ is the inclusion map, define $\rho' \colon S' \to F_I S'$ for all $s \in S$ and $x \in F_I S$ by
	\begin{align*}
		\rho'(s) = (F_I j \circ \rho)(s) &
			&
			\rho'(([\eta_I, \gamma_I] \circ F_I j)(x)) = F_I j(x).
	\end{align*}
	To see that this is well-defined, note that $[\eta_I, \gamma_I] \circ F_I j$ is injective.
	Moreover, if $s \in S$ and $x \in F_I S$ are such that
	\[
		s = ([\eta_I, \gamma_I] \circ F_I j)(x),
	\]
	we have
	\[
		({[\eta_I, \gamma_I]}^{-1} \circ j)(s) = F_I j(x).
	\]
	Because the inclusion map $j = {[\eta_I, \gamma_I]}^\rho$ by assumption, it follows that
	\[
		(F_I j \circ \rho)(s) = ({[\eta_I, \gamma_I]}^{-1} \circ j)(s) = F_I j(x).
	\]
	This completes the proof of $\rho'$ being well-defined.

	Note that $\rho$ is a subcoalgebra of ${[\eta_I, \gamma_I]}^{-1}$ via the inclusion map $j$ by assumption.
	The definition of $\rho'$ makes $\rho'$ also a subcoalgebra of ${[\eta_I, \gamma_I]}^{-1}$ via the inclusion map $S' \to TI$.
	This implies that it is recursive~\cite[Theorem~3.17]{adamek2007},\footnote{%
		As stated in the introduction of~\cite{adamek2007}, preservation of inverse images is implied by preservation of weak pullbacks.
	} and thus its inclusion map must be ${[\eta_I, \gamma_I]}^{\rho'} \colon S' \to TI$ by uniqueness.
\end{proof}

\elephant*
\begin{proof}
	Note that since $S$ is finite and $F$ preserves finite sets we have that $F(S + 1)$ is also finite.
	Together with the fact that $E$ is finite it follows that $E'$ is finite.
	Suppose $s_1, s_2 \in S$ are such that $(\beta_{E'} \circ \alpha_S)(s_1) = (\beta_{E'} \circ \alpha_S)(s_2)$.
	For all $s \in S$ we have
	\begin{align*}
		(o_\mathsf{t} \circ \alpha_S)(s) &
			= (o_\mathsf{t} \circ \mathfrak{e}_{\alpha_S(s)})(\square) \\
		&
			= (o_\mathsf{t} \circ [i_\mathsf{t}, \mathfrak{e}_{\alpha_S(s)}] \circ \kappa_2)(\square) \\
		&
			= (o_\mathsf{t} \circ {[i_\mathsf{t}, \mathfrak{e}_{\alpha_S(s)}]}^\sharp \circ \eta_{I + 1} \circ \kappa_2)(\square) \\
		&
			= (\beta_{E'} \circ \alpha_S)(s)((\eta_{I + 1} \circ \kappa_2)(\square)) \\
		&
			\phantom{{} = {}} \text{(definition of $\beta_{E'}$)},
	\end{align*}
	so
	\begin{align*}
		(o_\mathsf{t} \circ \alpha_S)(s_1) &
			= (\beta_{E'} \circ \alpha_S)(s_1)((\eta_{I + 1} \circ \kappa_2)(\square)) \\
		&
			= (\beta_{E'} \circ \alpha_S)(s_2)((\eta_{I + 1} \circ \kappa_2)(\square)) \\
		&
			= (o_\mathsf{t} \circ \alpha_S)(s_2).
	\end{align*}
	Furthermore, for all $s \in S$ we have
	\begin{align*}
		&
			\phantom{{} = {}} F[Ti_\mathsf{t} \circ j, T\mathfrak{e}_{\alpha_S(s)} \circ \eta_1] \\
		&
			= F[T[i_\mathsf{t}, \mathfrak{e}_{\alpha_S(s)}] \circ T\kappa_1 \circ j, T[i_\mathsf{t}, \mathfrak{e}_{\alpha_S(s)}] \circ T\kappa_2 \circ \eta_1] \\
		&
			= FT[i_\mathsf{t}, \mathfrak{e}_{\alpha_S(s)}] \circ F[T\kappa_1 \circ j, T\kappa_2 \circ \eta_1] \\
		&
			= FT[i_\mathsf{t}, \mathfrak{e}_{\alpha_S(s)}] \circ F[T\kappa_1 \circ j, \hat{\eta}_1] \\
		&
			\phantom{{} = {}} \text{(definition of $\hat{\eta}_1$)},
	\end{align*}
	so for all $s \in S$ and $x \in F(S + 1)$ we have
	\begin{align*}
		&
			\phantom{{} = {}} (\delta_\mathsf{t} \circ F[\alpha_S, \mathfrak{e}_{\alpha_S(s)}])(x) \\
		&
			= (\delta_\mathsf{t} \circ F[\delta_\mathsf{t}^* \circ Ti_\mathsf{t} \circ j, \delta_\mathsf{t}^* \circ T\mathfrak{e}_{\alpha_S(s)} \circ \eta_1])(x) \\
		&
			\phantom{{} = {}} \text{(definition of $\alpha_S$)} \\
		&
			= (\delta_\mathsf{t} \circ F\delta_\mathsf{t}^* \circ F[Ti_\mathsf{t} \circ j, T\mathfrak{e}_{\alpha_S(s)} \circ \eta_1])(x) \\
		&
			= (\delta_\mathsf{t}^* \circ \gamma_Q \circ F[Ti_\mathsf{t} \circ j, T\mathfrak{e}_{\alpha_S(s)} \circ \eta_1])(x) \\
		&
			\phantom{{} = {}} \text{($\delta_\mathsf{t}^*$ is an $F$-algebra homomorphism)} \\
		&
			= (\delta_\mathsf{t}^* \circ \gamma_Q \circ FT[i_\mathsf{t}, \mathfrak{e}_{\alpha_S(s)}] \circ F[T\kappa_1 \circ j, \hat{\eta}_1])(x) \\
		&
			\phantom{{} = {}} \text{(shown above)} \\
		&
			= (\delta_\mathsf{t}^* \circ T[i_\mathsf{t}, \mathfrak{e}_{\alpha_S(s)}] \circ \gamma_{I + 1} \circ F[T\kappa_1 \circ j, \hat{\eta}_1])(x) \\
		&
			= (\delta_\mathsf{t}^* \circ T[i_\mathsf{t}, \mathfrak{e}_{\alpha_S(s)}] \circ \gamma_{I + 1} \circ F[T\kappa_1 \circ j, \hat{\eta}_1] \circ \mathfrak{e}_x)(\square) \\
		&
			= (\delta_\mathsf{t}^* \circ T[i_\mathsf{t}, \mathfrak{e}_{\alpha_S(s)}] \circ c_x)(\square) \\
		&
			\phantom{{} = {}} \text{(definition of $c_x$)} \\
		&
			= ({[i_\mathsf{t}, \mathfrak{e}_{\alpha_S(s)}]}^\sharp \circ c_x)(\square) \\
		&
			= {[i_\mathsf{t}, \id_Q]}^\sharp((T(\id_I + \mathfrak{e}_{\alpha_S(s)}) \circ c_x)(\square)).
	\end{align*}
	Note that for all $s \in S$, $x \in F(S + 1)$ we have
	\begin{align*}
		&
			\phantom{{} = {}} \hat{\mu}_Q \circ T(\id_I + \mathfrak{e}_{(T(\id_I + \mathfrak{e}_{\alpha_S(s)}) \circ c_x)(\square)}) \\
		&
			= \hat{\mu}_Q \circ T(\id_I + (T(\id_I + \mathfrak{e}_{\alpha_S(s)}) \circ \mathfrak{e}_{c_x(\square)})) \\
		&
			\phantom{{} = {}} \text{\eqref{eq:one}} \\
		&
			= \hat{\mu}_Q \circ T(\id_I + (T(\id_I + \mathfrak{e}_{\alpha_S(s)}) \circ c_x)) \\
		&
			\phantom{{} = {}} \text{(definition of $\mathfrak{e}_{c_x(\square)}$)} \\
		&
			= \hat{\mu}_Q \circ T(\id_I + (T(\id_I + \mathfrak{e}_{\alpha_S(s)}))) \circ T(\id_I + c_x) \\
		&
			= T(\id_I + \mathfrak{e}_{\alpha_S(s)}) \circ \hat{\mu}_1 \circ T(\id_I + c_x),
	\end{align*}
	so for all $s \in S$, $x \in F(S + 1)$, and $e \in E$,
	\begin{align*}
		&
			\phantom{{} = {}} (\beta_E \circ \delta_\mathsf{t} \circ F(\alpha_S \circ [\id_S, \mathfrak{e}_s]))(x)(e) \\
		&
			= (\beta_E \circ \delta_\mathsf{t} \circ F[\alpha_S, \mathfrak{e}_{\alpha_S(s)}])(x)(e) \\
		&
			\phantom{{} = {}} \text{\eqref{eq:one}} \\
		&
			= (o_\mathsf{t} \circ {[i_\mathsf{t}, \mathfrak{e}_{(\delta_\mathsf{t} \circ F[\alpha_S, \mathfrak{e}_{\alpha_S(s)}])(x)}]}^\sharp)(e) \\
		&
			\phantom{{} = {}} \text{(definition of $\beta_E$)} \\
		&
			= (o_\mathsf{t} \circ {[i_\mathsf{t}, \mathfrak{e}_{{[i_\mathsf{t}, \id_Q]}^\sharp((T(\id_I + \mathfrak{e}_{\alpha_S(s)}) \circ c_x)(\square))}]}^\sharp)(e) \\
		&
			\phantom{{} = {}} \text{(shown earlier)} \\
		&
			= (o_\mathsf{t} \circ {[i_\mathsf{t}, \id_Q]}^\sharp \circ \hat{\mu}_Q \circ T(\id_I + \mathfrak{e}_{(T(\id_I + \mathfrak{e}_{\alpha_S(s)}) \circ c_x)(\square)}))(e) \\
		&
			\phantom{{} = {}} \text{(\autoref{lem:doublecontext})} \\
		&
			= (o_\mathsf{t} \circ {[i_\mathsf{t}, \id_Q]}^\sharp \circ T(\id_I + \mathfrak{e}_{\alpha_S(s)}) \circ \hat{\mu}_1 \circ T(\id_I + c_x))(e) \\
		&
			\phantom{{} = {}} \text{(shown above)} \\
		&
			= (o_\mathsf{t} \circ {[i_\mathsf{t}, \mathfrak{e}_{\alpha_S(s)}]}^\sharp \circ \hat{\mu}_1 \circ T(\id_I + c_x))(e) \\
		&
			= (\beta_{E'} \circ \alpha_S)(s)((\hat{\mu}_1 \circ T(\id_I + c_x))(e)) \\
		&
			\phantom{{} = {}} \text{(definition of $\beta_{E'}$)},
	\end{align*}
	and therefore for all $x \in F(S + 1)$ and $e \in E$,
	\begin{align*}
		&
			\phantom{{} = {}} (\beta_E \circ \delta_\mathsf{t} \circ F(\alpha_S \circ [\id_S, \mathfrak{e}_{s_1}]))(x)(e) \\
		&
			= (\beta_{E'} \circ \alpha_S)(s_1)((\hat{\mu}_1 \circ T(\id_I + c_x))(e)) \\
		&
			= (\beta_{E'} \circ \alpha_S)(s_2)((\hat{\mu}_1 \circ T(\id_I + c_x))(e)) &
			&
			\text{(assumption)} \\
		&
			= (\beta_E \circ \delta_\mathsf{t} \circ F(\alpha_S \circ [\id_S, \mathfrak{e}_{s_2}]))(x)(e).
	\end{align*}

	Thus, it follows from \autoref{lem:setlocalconsistency} that $(\alpha_S, \beta_{E'})$ is locally consistent w.r.t.\ $\beta_E$.
\end{proof}

\langequivrec*
\begin{proof}
	Suppose that for all recursive coalgebras $\rho \colon S \to F_I S$ such that $S$ is finite we have $o_\mathsf{t} \circ {[i_\mathsf{t}, \delta]}^\rho = o \circ {[i, \delta]}^\rho$.
	Given $t \in TI$, note that $(TI, [\eta_I, \gamma_I])$ is the initial algebra of functor $F_I$, which by being finitary is also the colimit of the initial sequence of $F_I$~\cite{adamek1974} and hence isomorphic to $(\bigcup_{n \in \N} F_I^n\emptyset, a)$ for an initial algebra structure $a \colon F_I\left(\bigcup_{n \in \N} F_I^n\emptyset\right) \to \bigcup_{n \in \N} F_I^n\emptyset$.
	Let $\phi \colon (TI, [\eta_I, \gamma_I]) \to (\bigcup_{n \in \N} F_I^n\emptyset, a)$ be the isomorphism.
	There exists $n \in \N$ such that $\phi(t) \in F_I^n\emptyset$.
	The set $F_I^n\emptyset$ is finite by $F_I$ preserving finite sets and the carrier of a recursive coalgebra $\rho \colon F_I^n\emptyset \to F_I^{n + 1}\emptyset$ by~\cite[Proposition~6]{capretta2006},
	with $a^\rho \colon F_I^n\emptyset \to \bigcup_{n \in \N} F_I^n\emptyset$ being the inclusion.
	Then $S = \{\phi^{-1}(x) \mid x \in F_I^n\emptyset\}$ is also finite and the carrier of a recursive coalgebra $\rho' \colon S \to F_I S$, with ${[\eta_I, \gamma_I]}^{\rho'} \colon S \to TI$ being the inclusion.
	Moreover, $t \in S$.
	Thus,
	\begin{align*}
		\lang_{\aut_\mathsf{t}}(t)
			&= (\lang_{\aut_\mathsf{t}} \circ {[\eta_I, \gamma_I]}^{\rho'})(t) \\
			&= (o_\mathsf{t} \circ i_\mathsf{t}^\sharp \circ {[\eta_I, \gamma_I]}^{\rho'})(t)
		        \tag{definition of $\lang_{\aut_\mathsf{t}}$} \\
		    &= (o_\mathsf{t} \circ {[i_\mathsf{t}, \delta_\mathsf{t}]}^{\rho'})(t)
		    &   \tag{$i_\mathsf{t}^\sharp$ is an $F_I$-algebra homomorphism} \\
		    &= \lang_{\aut_\mathsf{t}}^{\rho'}(t)
		    &   \tag{definition of $\lang_\aut^{\rho'}$} \\
		    &= \lang_{\aut}^{\rho'}(t)
		    &   \tag{assumption} \\
		    &= (o \circ {[i, \delta]}^{\rho'})(t)
		    &   \tag{definition of $\lang_{\aut}^{\rho'}$} \\
		    &= (o \circ i^\sharp \circ {[\eta_I, \gamma_I]}^{\rho'})(t)
		    &   \tag{$i^\sharp$ is an $F_I$-algebra homomorphism} \\
		    &= (\lang_{\aut} \circ {[\eta_I, \gamma_I]}^{\rho'})(t)
		    &   \tag{definition of $\lang_{\aut}$} \\
		    &= \lang_{\aut}(t).
	\end{align*}

	The converse follows from \autoref{prop:sublang}.
\end{proof}

The following useful lemma shows how contextual wrappers arise from recursive coalgebras.
\begin{lemma}\label{lem:rec-to-contextual}
	Let $\rho \colon S \rightarrow F_I S$ be a recursive coalgebra, and let $\iota^\triangleleft \colon S' \rightarrowtail TI$ be the image of the unique coalgebra-to-algebra map $\iota \colon S \rightarrow TI$ to the initial algebra. Then $S'$ is again a recursive coalgebra, $\iota^\triangleleft$ is the unique coalgebra-to-algebra morphism, and
	$\rch_\aut \circ \iota^\triangleleft = \alpha_{S'}$.
\end{lemma}
\begin{proof}
	We show the announced recursiveness property.
	Consider the factorisation $\iota^\triangleleft \circ \iota^\triangleright$
	of the unique coalgebra-to-algebra map $\iota$ in:
	\[
	\xymatrix{
		F_I S \ar[r]^{F_I \iota^\triangleright} 
			& F_I S' \ar@{>->}[r]^{F_I \iota^\triangleleft}  
			& F_I TI \ar[d]^{\cong} \\ 
		S \ar[u]^{\rho} \ar@{->>}[r]^{\iota^\triangleright} 
			& S' \ar@{>->}[r]^{\iota^ \triangleleft} \ar@{-->}[u]^d 
			& TI
	}
	\]
	Our functor $F_I$ preserves weak pullbacks (since $F$ does), and hence monos in particular. Thus $F_I(\iota^\triangleleft)$ is monic. Since the initial algebra
	is an isomorphism, the composition with its structure map is monic as well
	and we get the map $d$ by the fill-in property of the factorisation system.

	Finally, again since $F_I$ preserves weak pullbacks it preserves intersections; and as a consequence, the coalgebra $d \colon S' \rightarrow F_I(S')$ is recursive by~\cite[Corollary 5.6]{AdamekMM20}.
\end{proof}

\end{document}

%% file: columns.tex
\newcolumntype{C}[1]{>{\centering\arraybackslash$}m{#1}<{$}}
\newcolumntype{L}[1]{>{\raggedleft\arraybackslash$}m{#1}<{$}}
\newcolumntype{R}[1]{>{\raggedright\arraybackslash$}m{#1}<{$}}